\documentclass[1p,authoryear,11pt,hidelinks]{elsarticle}
\usepackage{hyperref}
\usepackage{amssymb}
\usepackage{eurosym}
\usepackage{amsfonts}
\usepackage{amsmath}
\numberwithin{equation}{section}
\usepackage{amsthm}
\usepackage{verbatim}
\usepackage{graphicx}
\usepackage{caption}
\usepackage{subcaption}
\usepackage{natbib}
\usepackage[inline]{enumitem}
\journal{European Journal of Operational Research}

\newtheorem{theorem}{Theorem}[section]

\newtheorem{definition}[theorem]{Definition}

\newtheorem{lemma}[theorem]{Lemma}
\newtheorem{proposition}[theorem]{Proposition}

\theoremstyle{plain}

\makeatletter
\newcommand*\rel@kern[1]{\kern#1\dimexpr\macc@kerna}
\newcommand*\widebar[1]{%
  \begingroup
  \def\mathaccent##1##2{%
    \rel@kern{0.8}%
    \overline{\rel@kern{-0.8}\macc@nucleus\rel@kern{0.2}}%
    \rel@kern{-0.2}%
  }%
  \macc@depth\@ne
  \let\math@bgroup\@empty \let\math@egroup\macc@set@skewchar
  \mathsurround\z@ \frozen@everymath{\mathgroup\macc@group\relax}%
  \macc@set@skewchar\relax
  \let\mathaccentV\macc@nested@a
  \macc@nested@a\relax111{#1}%
  \endgroup
}
\makeatother
\newcommand*{\bigs}[1]{\scalebox{1.2}{\ensuremath#1}}

\makeatletter
\def\ps@pprintTitle{%
     \let\@oddhead\@empty
		 \let\@evenhead\@empty
     \def\@oddfoot
       {\hbox to \textwidth%
        {\ifnopreprintline\relax\else
        \@myfooterfont%
         \ifx\@elsarticlemyfooteralign\@elsarticlemyfooteraligncenter%
           \hfil\@elsarticlemyfooter\hfil%
         \else%
         \ifx\@elsarticlemyfooteralign\@elsarticlemyfooteralignleft%
           \@elsarticlemyfooter\hfill{}%
         \else%
         \ifx\@elsarticlemyfooteralign\@elsarticlemyfooteralignright%
           {}\hfill\@elsarticlemyfooter%
         \else%
            \begin{minipage}[t]{\textwidth}   Preprint accepted in \ifx\@journal\@empty%
                 Elsevier%
            \else\@journal\fi\hfill December 23, 2019\\
						Published version available at \rurl{10.1016/j.ejor.2019.12.032}\medskip\\ 
						\begin{minipage}[b]{0.7\textwidth}This work is licensed under a Creative Commons\\ 
			      Attribution-NonCommercial-NoDerivatives\\
						4.0 International License\end{minipage} 
							\hfill\includegraphics[width=0.2\textwidth]{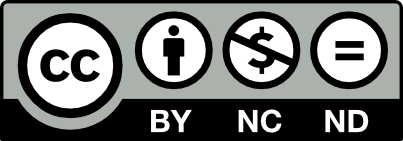}
						\end{minipage}\fi%
         \fi%
         \fi%
         \fi%
         }
       }%
     \let\@evenfoot\@oddfoot}
\makeatother
\newcommand\rurl[1]{%
  \href{http://doi.org/#1}{\nolinkurl{#1}}%
}

\begin{document}
\setlength{\footskip}{30pt}

\begin{frontmatter}

\title{Simple Explicit Formula for Near-Optimal Stochastic Lifestyling}

\author[cass]{Ale\v{s} \v{C}ern\'{y}}\corref{ca}
\cortext[ca]{\textit{Corresponding author email:} \href{mailto:ales.cerny.1@city.ac.uk}{{\ttfamily ales.cerny.1@city.ac.uk}}}
\author[uniba]{Igor Melicher\v{c}\'{\i}k}
\ead{igor.melichercik@fmph.uniba.sk}
\address[cass]{Cass Business School, City, University of London, 106 Bunhill Row, London EC1Y 8TZ, UK}
\address[uniba]{Department of Applied Mathematics and Statistics, Comenius University Bratislava, 84248 Bratislava, Slovakia}

\begin{abstract}
In life-cycle economics, the Samuelson paradigm \citep{samuelson.69} states that the optimal investment is in constant proportions out of lifetime wealth composed of current savings and the present value of future income. It is well known that in the presence of credit constraints this paradigm no longer applies. Instead, optimal life-cycle investment gives rise to so-called stochastic lifestyling \citep{cairns.al.06}, whereby for low levels of accumulated capital it is optimal to invest fully in stocks and then gradually switch to safer assets as the level of savings increases. In stochastic lifestyling not only does the ratio between risky and safe assets change but also the mix of risky assets varies over time. While the existing literature relies on complex numerical algorithms to quantify optimal lifestyling, the present paper provides a simple formula that captures the main essence of the lifestyling effect with remarkable accuracy.
\end{abstract}

\begin{keyword}
finance\sep optimal investment\sep stochastic lifestyling\sep Samuelson paradigm\sep power utility
\MSC[2010] 90C20\sep 90C39\sep 35K55\sep 49J20
\end{keyword}

\end{frontmatter}

\section{Introduction}

Operations research has analysed pension finance from two angles. The first looks at practical methodology for  asset--liability management of a pension scheme as a whole (\citealp{sodhi.05}; \citealp{mulvey.al.08}). The second seeks to characterize the optimal mix of risky and risk-free investments for individual members of a pension scheme as they progress from early working life to retirement (\citealp{cairns.al.06}; \citealp{zhang.ewald.10}). This second stream is informed by and linked to a wider literature on optimal investment and consumption with constraints (\citealp{zariphopoulou.94}; \citealp{vila.zariphopoulou.97}; \citealp{xia.11}; \citealp{nutz.12.mf}; \citealp{kilianova.sevcovic.13}).

In contrast to the considerable mathematical and numerical sophistication needed to arrive at optimal pension portfolios, there is notable absence of portfolio rules that are simple to implement and yet do not compromise welfare of investors. 
The practical need for such rules is significant but this demand has not been met by academia, despite five decades of research. In an isolated contribution, \citet{ayres.nalebuff.13} propose simple heuristic rules for life-cycle portfolio allocation and evaluate their welfare, without analyzing their optimality. This paper offers an insight how one may bridge the gap between optimality and ease of implementation.%
\footnote{Thanks to their tractability our results have been adopted by Allianz in a spreadsheet modeller available to individual pension
account clients in Slovakia.} 

Consider a model with $d$ risky assets whose dynamics are given by the stochastic differential equation (SDE)
\begin{equation}
\frac{dS_{t}}{S_{t}}=\mu dt+\sigma dB_{t},  \label{lognormal}
\end{equation}
where $B$ are $d$ uncorrelated Brownian motions, $\mu \in \mathbb{R}^{d}$, and $\Sigma =\sigma \sigma ^{\top }\in \mathbb{R}^{d\times d}$ is regular. Assume further that there is a risk-free asset with value $S^{0}=e^{rt}$. An individual who starts working at time $0$ and retires at time $T$ makes pension contributions at the deterministic rate $y_{t}$ per unit of time. The task of the pension fund manager is to invest these contributions on behalf of the individual so as to maximize the expected utility of the terminal value of the pension plan. To aid tractability, it is customary to consider utility functions of the form 
\begin{equation*}
U_\gamma(x)=\frac{x^{1-\gamma }}{1-\gamma },\quad \gamma >0,\gamma \neq 1.
\end{equation*}
The analysis can be extended to $\gamma =1$ with $U_1(x)=\ln x$ and we will do so in due course. 

We seek the optimal investment plan $\pi^*$ that solves 
\begin{subequations}\label{eq:intro}
\begin{align}
\pi^*&=\operatorname*{arg\,max}_{\pi\geq 0,\ \pi \mathbf{1}\leq 1}E\left[U_\gamma(W_T)\right] \text{ subject to}\\
dW_t&=\left( rW_t+y_t\right) dt+\pi_t W_t\left( \frac{dS_t}{S_t}-r\mathbf{1}dt\right). \label{eq:dW}
\end{align}
\end{subequations}
Here $W_t$ denotes accumulated savings and $\pi$ represents the proportions invested in the risky assets.%
\footnote{By convention, $\pi$ is a row vector while $S$, $\mu$, and $\mathbf{1}$ are column vectors.}
Parameter $\gamma$ captures the risk-aversion of the individual account holder. The restrictions imposed on $\pi$ reflect typical institutional constraints faced by pension funds. In addition to  shortsale constraints on risky assets, $\pi\geq 0$, there is a credit constraint that prevents the fund manager from borrowing against the value of future contributions, $\pi \mathbf{1}\leq 1$.

It is well known that without constraints on $\pi$ and without contributions ($y_t=0$) the optimal investment strategy is given by
\begin{equation}\label{eq: pihat1}
\pi^*=\frac{(\mu -r\mathbf{1})^{\top }\Sigma ^{-1}}{\gamma} = 
\underset{\pi\in\mathbb{R}^d}{\arg\max }\ \pi(\mu -r\mathbf{1})-\frac{\gamma}{2}\pi\Sigma\pi^\top.  
\end{equation}
In the context of the optimization problem \eqref{eq:intro}, one is thus lead to consider a heuristic fixed proportions strategy
\begin{equation}
\pi ^{(1)}=\text{ }\underset{\pi \geq 0,\pi \mathbf{1}\leq 1}{\arg \max }\
\pi (\mu -r\mathbf{1})-\frac{\gamma }{2}\pi \Sigma \pi ^{\top }.  \label{eq: pi(1)}
\end{equation}
Suppose the weights in \eqref{eq: pihat1} are strictly positive. Taken as a function of risk aversion $\gamma$, the optimal weights $\pi
^{(1)}$ are no longer equal to the risky mix from \eqref{eq: pihat1} adjusted for the leverage constraint $\pi\mathbf{1}\leq 1$, as given by the formula 
\begin{equation}
\pi ^{(0)}=\frac{(\mu -r\mathbf{1})^{\top }\Sigma ^{-1}}{\max \left( (\mu -r\mathbf{1})^{\top }\Sigma ^{-1}\mathbf{1},\gamma \right)}.  \label{eq: pi(0)}
\end{equation}
Instead, for low levels of the risk aversion parameter $\gamma $ the \emph{relative weights} in $\pi ^{(1)}$ change in a way that entails substitution towards the riskier assets as $\gamma $ decreases.

One might reasonably expect that strategy \eqref{eq: pi(1)} would provide satisfactory heu\-ris\-tic approximation of the fully optimal investment strategy. However, numerical experiments reveal that the character of the optimal investment changes more dramatically than suggested by equation \eqref{eq: pi(1)}. Simulations capture a phenomenon known in pension finance as stochastic lifestyling, a term coined by \citet{cairns.al.06}, whereby it is optimal early on to invest the accumulated savings in stocks and then gradually switch the investment into bonds and safe deposits as the retirement approaches and the total amount of savings increases. Thus the optimal strategy behaves \emph{as if} the risk-aversion coefficient were lower for low levels of accumulated funds.

Because the fully optimal strategy $\pi^*$ in \eqref{eq:intro} has to be computed numerically by \emph{dynamic} programming and because it is a non-linear function of both time $t$ and the accumulated savings $W_{t}$, at first sight it is difficult to see how one can characterize the lifestyling effect explicitly. In this paper we point out that there is an excellent heuristic approximation of the lifestyling effect, given by a formula that is no less explicit than equation \eqref{eq: pi(1)}.

To arrive at the correct lifestyling formula, one must adopt Samuelson's view of the investment weights \eqref{eq: pihat1}. When the individual savings plan can borrow as well as invest at the risk-free rate $r$ \citet{samuelson.69}, 
and more explicitly \citet{hakansson.70}, have pointed out that the presence of contributions does not affect the constant proportions strategy \eqref{eq: pihat1} provided that the risky investment is made out of lifetime pension wealth 
\begin{equation*}
\widebar{W}_t = W_t + \mathrm{PV}_t,
\end{equation*} 
where $\mathrm{PV}_t$ is the present value of all future pension contributions as of time $t$.

If we denote by $\widebar{\pi}_t$ the proportions of risky investment out of lifetime pension wealth $\widebar{W}_t$, the
credit constraint $\pi_t\mathbf{1}\leq 1$ is transformed to $\widebar{\pi }_{t}\mathbf{1}\leq \alpha _{t}$, where 
\begin{equation}\label{eq: alpha}
\alpha _{t}=\frac{W_{t}}{\widebar{W}_t}  
\end{equation}
is the ratio of the already accumulated savings to the entire lifetime pension capital. Observe that in the Samuelson world the heuristic strategy $\pi ^{(1)}$ corresponds to 
\begin{equation*}
\widebar{\pi }^{(1)}\left( \alpha _{t}\right) =\alpha _{t}\pi ^{(1)}.
\end{equation*}
Observe also that if the sum of weights $\pi ^{(1)}\mathbf{1}$ is strictly less than 1 then the sum of weights in 
$\widebar{\pi }^{(1)}\left( \alpha_{t}\right)$ will be strictly less than $\alpha _{t}$ for all $\alpha_{t}\in (0,1)$ which is unlikely to be optimal. We therefore also consider a modified heuristic 
\begin{equation*}
\widebar{\pi }^{(2)}\left( \alpha _{t}\right) =\min \left( \frac{\alpha_{t}}{\pi ^{(1)}\mathbf{1}},1\right) \pi ^{(1)},  
\end{equation*}
that corresponds to cash-in-hand investment proportions
\begin{equation}\label{eq: pi(2)}
\pi ^{(2)}\left( \alpha _{t}\right) =\frac{\pi ^{(1)}}{\max \left( \pi
^{(1)}\mathbf{1},\alpha _{t}\right) }.  
\end{equation}

However, the key breakthrough of this paper is achieved by formulating a heuristic strategy directly in the Samuelson world, in the form 
\begin{equation*}
\widebar{\pi }^{(3)}(\alpha _{t})
=\underset{\pi \geq 0,\pi \mathbf{1}\leq\alpha _{t}}{\arg \max }\ \pi (\mu -r\mathbf{1})-\frac{\gamma }{2}\pi \Sigma \pi^{\top },  
\end{equation*}
which, when expressed as proportions out of accumulated savings $W_t$, yields
\begin{equation}\label{eq: pi(3)intro}
\pi ^{(3)}\left( \alpha _{t}\right) =\frac{\widebar{\pi }^{(3)}(\alpha_{t})}{\alpha _{t}}
=\underset{\pi \geq 0,\pi \mathbf{1}\leq 1}{\arg \max }\ \pi (\mu -r\mathbf{1})-\frac{\alpha_t\gamma }{2}\pi \Sigma \pi^{\top }.  
\end{equation}

We show that, unlike $\pi ^{(1)}$ and $\pi ^{(2)}(\alpha _{t})$, the strategy $\pi ^{(3)}(\alpha _{t})$ is an excellent approximation to the fully optimal strategy and can therefore serve as a simple rule of thumb for pension plan providers who wish to offer a choice of lifestyling strategies to their clients, while also specifying the sense in which such lifestyling is optimal. To reduce the barriers to application further, we analyze the explicit dependence of $\pi ^{(3)}$ on $\alpha _{t}$ for a given set of binding constraints. For example, assuming that the constraints $\pi \geq 0$ are not binding, the near-optimal strategy $\pi ^{(3)}$ is of the form
\begin{equation}\label{eq: pi2explicit}
\pi ^{(3)}(\alpha _{t})=\frac{(\mu -r\mathbf{1})^\top\Sigma^{-1}}{\gamma \alpha _{t}}
+\frac{\mathbf{1}^{\top }\Sigma ^{-1}}{\mathbf{1}^{\top }\Sigma ^{-1}\mathbf{1}}
\min \left(1-\frac{(\mu -r\mathbf{1})^\top\Sigma^{-1}\mathbf{1}}{\gamma \alpha _{t}},0\right).  
\end{equation}
Note that the non-negativity constraint will become binding for $\alpha _{t}$ small enough, at which point, for typical parameter values, 
the formula directs all accumulated savings to be invested in stocks. Interestingly, 
$\mathbf{1}^{\top }\Sigma ^{-1}/\mathbf{1}^{\top }\Sigma^{-1}\mathbf{1}$ is the classical Markowitz minimum variance portfolio.

Formula \eqref{eq: pi2explicit} captures the main essence of the lifestyling effect, representing in a nutshell the main conceptual
contribution of our paper. It not only shows the change in portfolio composition as a function of $\alpha _{t}$ for fixed risk 
aversion, but it also neatly demonstrates that the portfolio composition will change with decreasing $\gamma $ when there are no future contributions to consider ($\alpha _{t}=1$). According to the formula, the near-optimal investment proportions do behave as if the risk aversion were lower for low levels of accumulated funds, with effective risk aversion equal to $\alpha _{t}\gamma $.

The article is organized as follows. Section~\ref{S:theory} introduces what we call the `Samuelson transform', linking a model with gradual contributions to an equivalent model where all capital is paid up-front but there are additional constraints on how the capital can be invested. We review the mathematical theory guaranteeing existence of an optimal strategy in the world with contributions and via the Samuelson link also in the world without contributions but with investment constraints. In Section~\ref{S:analysis} we provide economic analysis of the competing strategies, both in terms of welfare impact and portfolio weights. We close this section with a thorough robustness analysis. Section~\ref{S:conclusions} concludes.

\section{Theory}\label{S:theory}
\subsection{Samuelson transform \label{SS:Samuelson}}
We denote by $Y_{t}=\int_{0}^{t}y(u)du$ the cumulative pension contribution up to and including time $t$. Function $y$ is assumed to be deterministic, non-negative, and integrable on $[0,T]$. The price process of all assets, including the risk-free asset, is denoted by 
$S = (S^{0},S^{1:d})$. We assume $S^{1:d}$ is a geometric Brownian motion with drift as described in equation \eqref{lognormal}, 
while $S_{t}^{0}=e^{rt}$ represents a bank account with risk-free deposit rate $r.$ Risk-free borrowing is excluded.

The process 
\begin{equation*}
\mathrm{PV}_{t}=\int_{t}^{T}e^{-r(u-t)}dY_{u},
\end{equation*}
is the present value at time $t$ of all contributions in the period $(t,T]$.

\begin{definition}\label{def: sf} 
We say that $\varphi$ is a self-financing strategy for price process $S$ and cumulative contributions $Y$, writing 
$\varphi\in\Theta (S,Y)$, if $\varphi$ is predictable, $S$--integrable, and 
\begin{equation*}
\varphi _{0}S_{0}+\int_{0}^{t}\varphi _{u}dS_{u}+Y_{t}=\varphi _{t}S_{t}.
\end{equation*}
We denote by $\Theta _{x}(S,Y)$ the set of all self-financing strategies with initial capital $x$, 
\begin{equation*}
\Theta _{x}(S,Y)=\{\varphi \in \Theta (S,Y):\varphi _{0}S_{0}=x\}.
\end{equation*}
\end{definition}

Consider the following transformation of trading strategies $\varphi \mapsto 
\widebar{\varphi }$: 
\begin{align}
\widebar{\varphi}_{t}^{1:d}& =\varphi_{t}^{1:d},  \label{samuelson1} \\
\widebar{\varphi}_{t}^{0}& =\varphi_{t}^{0}+e^{-rt}\mathrm{PV}_{t}.
\label{samuelson2}
\end{align}%
We call (\ref{samuelson1}--\ref{samuelson2}) the \emph{Samuelson transform}. Using the numeraire change technique of \citet{geman.al.95}
it is readily seen that the Samuelson transform\emph{\ }is a one-to-one mapping between $\Theta _{x}(S,Y)$ and 
$\Theta _{x+\mathrm{PV}_{0}}(S,0)$.

We can now turn our attention to a situation where borrowing against future contributions is no longer possible.

\begin{definition}
\label{d2} Consider an arbitrary self-financing strategy $\varphi \in \Theta_{x}(S,Y)$ with an arbitrary contribution process $Y.$ Assume that $\varphi\geq 0$ and $S\geq 0$. We define the vector of proportions, $\pi (\varphi )$, invested in available risky assets by 
\begin{equation*}
\pi _{i}(\varphi )=\frac{\varphi ^{i}S^{i}}{\varphi S},\qquad \qquad i\in\{1,\ldots,d\},  
\end{equation*}
using the convention $0/0=0$.
\end{definition}

\begin{proposition}\label{p2} 
Suppose $S\geq 0$. The Samuelson transform is a one-to-one mapping between 
$\mathcal{A}_{x}=\left\{ \varphi \in \Theta _{x}(S,Y):\pi (\varphi )\geq 0,\pi (\varphi )\mathbf{1}\leq 1\right\}$, and
\begin{equation}\label{eq: C3}
\widebar{\mathcal{A}}_{x+\mathrm{PV}_{0}}=\{\widebar{\varphi }\in \Theta _{x+\mathrm{PV}_{0}}(S,0)
:\pi (\widebar{\varphi })\geq 0,\pi (\widebar{\varphi})\mathbf{1}\leq 1-\mathrm{PV}/\widebar{\varphi }S\}.  
\end{equation}
\end{proposition}

\begin{proof}$\pi (\varphi )\geq 0\,\wedge\,\pi (\varphi )\mathbf{1}\leq 1\iff\varphi ^{0}S^{0}\geq 0\,\wedge\,\varphi ^{1:d}\geq 0
\iff\widebar{\varphi}^{0}S^{0}\geq~\mathrm{PV}\wedge\,\widebar{\varphi }^{1:d}\geq 0
\iff \pi (\widebar{\varphi })\geq 0\,\wedge\, \pi (\widebar{\varphi })\mathbf{1}\leq 1-\mathrm{PV}/\widebar{\varphi }S$. 
\end{proof}

Proposition~\ref{p2} clarifies the link between the classical Samuelson paradigm and the situation where the risk-free borrowing against future contributions is precluded. While in the classical case the sum of risky proportions is unconstrained, there is now in \eqref{eq: C3} a stochastic constraint on the total proportion invested in the risky assets. The risky proportion must not exceed $1-\mathrm{PV}/\widebar{\varphi }S$ in Samuelson's world without contributions. In economic terms, risky investment can only be financed from past contributions and from past capital gains. Below, we investigate how this constraint influences the leverage and the relative proportions invested in risky assets.

\subsection{Hamilton--Jacobi--Bellman equations}

In this subsection we relate the optimal investment strategy to the solutions of two Hamilton--Jacobi--Bellman (HJB) equations. The twin representation turns out to be important in the proof of existence and uniqueness (Subsection~\ref{SS:existence}) and in the proof of optimality (Subsection~\ref{SS:optimality}) but most importantly it provides economic motivation for the near-optimal strategy 
(Subsection~\ref{SS:pi3}).

For the sake of brevity, hereafter we consider a constant contribution rate $y$. We begin by writing out formally the partial differential equation (PDE) in the world with contributions, \vspace{-0.15in}
\begin{subequations}
\label{eq: HJB IVP}
\begin{align}
0={}& \sup_{\pi \geq 0,\pi \mathbf{1\leq }1}v_{t}+v_{x}(y+\left( r+\pi (\mu
-r\mathbf{1})\right) x)+\frac{x^{2}}{2}v_{xx}\pi \Sigma \pi ^{\top },  \label{eq: HJB}
\\
v(T,x)={}& \frac{x^{1-\gamma }}{1-\gamma}.  \label{eq: HJB_boundary}
\end{align}
\end{subequations}
The terms standing by $v_{x}$ and $v_{xx}$ originate from the dynamics of accumulated savings $W$ in \eqref{eq:dW}.
In Samuelson's world without contributions the corresponding HJB equation reads\vspace{-10pt}%
\begin{subequations}
\label{eq: HJBbar IVP}
\begin{align}
0={}& \sup_{\widebar{\pi }\geq 0,\widebar{\pi }\mathbf{1\leq }1-\mathrm{PV}_t/\widebar{x}}\widebar{v}_{t}+\widebar{x}\widebar{v}_{\widebar{x}}(r+\widebar{\pi }(\mu -r\mathbf{1}))+\frac{\widebar{x}^{2}}{2}\widebar{v}_{\widebar{x}%
\widebar{x}}\widebar{\pi }\Sigma \widebar{\pi }^{\top },
\label{eq: HJBbar} \\
\widebar{v}(T,\widebar{x})={}& \frac{\widebar{x}^{1-\gamma }}{1-\gamma},
\label{eq: HJBbar_boundary}
\end{align}%
\end{subequations}
corresponding to lifetime pension wealth dynamics 
\begin{equation}\label{eq: dWbar}
d\widebar{W}_t=r\widebar{W}_tdt+\widebar{\pi }_t\widebar{W}_t\left( \frac{dS_t}{S_t}-r\mathbf{1}dt\right).  
\end{equation}

Similarly, the value function corresponding to the heuristic strategy $\pi^{(i)}$ for $i\in\{0,1,2,3\}$ in the world with contributions is formally given as a solution of 
\begin{subequations}
\label{eq: v(i) IVP}
\begin{align}
0={}& v_{t}^{(i)}+v_{x}^{(i)}\left( y+\left( r+\pi ^{(i)}(\mu -r\mathbf{1})x\right)
\right) +\frac{x^{2}}{2}v_{xx}\pi ^{(i)}\Sigma \pi ^{(i)\top },
\label{eq:v^(i) PDE} \\
v^{(i)}(T,x)={}& \frac{x^{1-\gamma }}{1-\gamma},  \label{eq: v^(i) BC}
\end{align}
\end{subequations}
where $\pi ^{(i)}$ is taken to be a fixed function of $(t,x)$ as indicated in the introduction. In the Samuelson world, one obtains an
analogous PDE for the strategies $\widebar{\pi }^{(i)}$, \vspace{-10pt}
\begin{subequations}
\label{eq: vbar(i) IVP}
\begin{align}
0={}& \widebar{v}_{t}^{(i)}
+\widebar{x}\widebar{v}_{\widebar{x}}^{(i)}(r+\widebar{\pi }^{(i)}(\mu -r\mathbf{1}))
+\frac{\widebar{x}^{2}}{2}\widebar{v}_{\widebar{x}\widebar{x}}^{(i)}\widebar{\pi }^{(i)}\Sigma \widebar{\pi }^{(i)\top },  
\label{eq:vbar^(i) PDE} \\
\widebar{v}^{(i)}(T,\widebar{x})={}& \frac{\widebar{x}^{1-\gamma }}{1-\gamma}. \label{eq: vbar^(i) BC}
\end{align}
\end{subequations}

The two sets of equations are equivalent in the sense that every $\mathcal{C}^{1,2}$ solution of the initial value problem 
\eqref{eq: HJB IVP} generates a $\mathcal{C}^{1,2}$ solution of \eqref{eq: HJBbar IVP} via transformation 
$\widebar{v}(t,\widebar{x})=v(t,\widebar{x}-\mathrm{PV}_t).$ Conversely, any $\mathcal{C}^{1,2}$ solution of \eqref{eq: HJBbar IVP} gives rise to a $\mathcal{C}^{1,2}$ solution of \eqref{eq: HJB IVP} through $v(t,x)=\widebar{v}(t,x+\mathrm{PV}_t)$. The same correspondence holds between \eqref{eq: v(i) IVP} and \eqref{eq: vbar(i) IVP}.

If, for the time being, we accept as given that \eqref{eq: HJB IVP}, resp. \eqref{eq: HJBbar IVP}, admit optimal controls $\pi^*$, resp. $\widebar{\pi}^*$, then there is also a relationship between \eqref{eq: HJB IVP} and \eqref{eq: v(i) IVP} to the extent 
that if one substitutes $\pi^*$ for $\pi^{(i)}$ in \eqref{eq: v(i) IVP} one obtains a solution of \eqref{eq: HJB IVP}. The same correspondence holds between \eqref{eq: HJBbar IVP} and \eqref{eq: vbar(i) IVP} on replacing $\widebar{\pi}^{(i)}$ with $\widebar{\pi}^*$. 

Before we examine the optimal controls it is helpful to associate a
coefficient of risk aversion to each indirect utility,
\begin{align}
R(t,x)& =-\frac{xv_{xx}(t,x)}{v_{x}(t,x)},  \label{eq: R} \\
\widebar{R}(t,\widebar{x})& =-\frac{\widebar{x}\widebar{v}_{\widebar{x}\widebar{x}}(t,\widebar{x})}{\widebar{v}_{\widebar{x}}(t,\widebar{x})}.
\label{eq: Rbar}
\end{align}
The optimal portfolio strategy is related to the following deterministic mean-variance utility 
$f:[0,\infty )\times (0,\infty )\rightarrow \mathbb{R}$, with risky investment constraint $\alpha $ and risk aversion $\rho$,  
\begin{equation}\label{eq: MVmax}
f(\alpha ,\rho )=\sup_{\pi \geq 0,\pi \mathbf{1\leq }\alpha }\pi (\mu -r\mathbf{1})-\frac{\rho }{2}\pi \Sigma \pi ^{\top }.  
\end{equation}
Due to strict convexity in $\pi $ and compactness of the optimization region
there is a unique optimizer in the deterministic problem \eqref{eq: MVmax} which we denote $\hat{\pi}(\alpha ,\rho )$,
\begin{equation}\label{eq: pihat def}
\hat{\pi}(\alpha ,\rho )=\underset{\pi \geq 0,\pi \mathbf{1}\leq\alpha}{\arg \max }\ 
\pi (\mu -r\mathbf{1})-\frac{\rho }{2}\pi \Sigma \pi ^{\top }.  
\end{equation}
We note for future use that $\hat{\pi}(\alpha ,\rho )$ is self-similar, that is, for $\alpha >0$ one has
\begin{equation}\label{eq: self similar}
\hat{\pi}(\alpha ,\rho )=\alpha \hat{\pi}(1,\alpha \rho ),
\end{equation}
with the convention $0\times \infty =0$.

Using the newly established notation the formal optimal controls in \eqref{eq: HJB IVP} and \eqref{eq: HJBbar IVP} can be written as 
\vspace{-0.14in}
\begin{subequations}
\label{eq: pi* pibar*}
\begin{align}
\pi ^{\ast }(t,x) ={}&\hat{\pi}(1,R(t,x)),  \label{eq: pi*} \\
\widebar{\pi }^{\ast }(t,\widebar{x}) ={}&\hat{\pi}(1-\mathrm{PV}_t/\widebar{x},\widebar{R}(t,\widebar{x})). \label{eq: pibar*}
\end{align}
\end{subequations}
Furthermore, the self-similarity of $\hat{\pi}(\alpha ,\rho )$ yields 
\begin{align*}
\pi ^{\ast }(t,x) &=\left( 1+\mathrm{PV}_t/x\right) \widebar{\pi }^{\ast
}(t,x+\mathrm{PV}_t), \\
\widebar{\pi }^{\ast }(t,\widebar{x}) 
&=\left( 1-\mathrm{PV}_t/\widebar{x}\right) \pi ^{\ast }(t,\widebar{x}-\mathrm{PV}_t).
\end{align*}
Economically this is no surprise in the light of our analysis in Subsection~\ref{SS:Samuelson}.

\subsection{Existence and uniqueness}\label{SS:existence} 
The advantage of the world with contributions is
that it measures investment in natural units -- out of accumulated funds. In addition, it is mathematically better behaved in that it can be transformed to a strictly parabolic quasilinear PDE whose properties, albeit mathematically involved, are well understood in specialist literature.%
\footnote{A related constrained optimization problem is studied in \citet{vila.zariphopoulou.97}. Their proofs make it clear that a rigorous 
mathematical treatment of the problem is technically demanding. We follow an alternative line of attack proposed in 
\citet{kilianova.sevcovic.13} that allows us to condense the technical arguments considerably.}
\begin{theorem}
\label{th: main}Under the assumption 
\begin{equation}
\mu _{i}>r,\text{\qquad for some }i\in \{1,\ldots, d\},  \label{eq: mu > r}
\end{equation}
the inital value problems (\ref{eq: HJB IVP}--\ref{eq: vbar(i) IVP}) have a unique classical solution belonging to 
$\mathcal{C}^{1,2}([0,T]\times (0,\infty ))$. The corresponding maximizers $\pi ^{\ast }(t,x)$ and $\widebar{\pi }^{\ast }(t,x)$ from 
\eqref{eq: pi* pibar*} have the property that $x\pi ^{\ast }(t,x)$, resp. $x\widebar{\pi }^{\ast }(t,x)$, is locally Lipschitz--con\-ti\-nuous in $x$, uniformly in $t$, on $[0,T]\times \lbrack 0,\infty )$.
\end{theorem}

\begin{proof} 1) The difficult part is to reformulate the problem into a form where strict parabolicity can be established. We follow the strategy of \citet{kilianova.sevcovic.13} whose key result is summarized in Proposition~\ref{prop: KS13}. One begins with equation 
\eqref{eq: HJBu} formally obtained from \eqref{eq: HJB} by a logarithmic transformation $x\rightarrow e^{z},v(t,x)\rightarrow u(t,z)$.
Momentarily granting the assumptions of Proposition~\ref{prop: KS13} one establishes the existence and properties of an auxiliary function 
$\rho(t,z) $ from \eqref{eq: PDE_rho}. Subsequently, from $\rho $ one constructs via \eqref{eq: PDE_u_from_rho} $u$ as a solution of 
\eqref{eq: HJBu} with a further property $1-u_{zz}/u_{z}=\rho $. Therefore, the indirect risk aversion coefficient 
$R(t,x)=-xv_{xx}/{v_{x}} = \rho (t,\ln x)$ belongs to $\mathcal{C}^{1,2}([0,T]\times (0,\infty ))$.

2) It is now readily seen that $v(t,x)=u(t,\ln x)$ is a unique classical solution of the HJB equation \eqref{eq: HJB} and likewise 
$\widebar{v}(t,\widebar{x})=v(t,\widebar{x}-\mathrm{PV}_t)$ is a unique classical solution of the HJB equation \eqref{eq: HJBbar}.

3) To invoke Proposition~\ref{prop: KS13}, it remains to prove that under the assumptions of Theorem~\ref{th: main} function $g$, 
\begin{equation*}
g\left( \rho \right) =f(1,\rho ) = \sup_{\pi \geq 0,\pi \mathbf{1\leq }1}\pi (\mu -r\mathbf{1})
-\frac{\rho }{2}\pi \Sigma \pi ^{\top },
\end{equation*}%
possesses locally Lipschitz-continuous derivative with the property 
\begin{equation}
0<\inf_{\rho \in (0,\gamma ]}-g^{\prime }(\rho )\leq \sup_{\rho \in
(0,\gamma ]}-g^{\prime }(\rho )<\infty .  \label{eq: property g0}
\end{equation}
Since the region $A=\{\pi \in \mathbb{R}^{d}:\pi \geq 0,\pi \mathbf{1}\leq 1\}$ is compact, one has 
\begin{equation}\label{eq: sup pi in A}
\sup_{\pi \in A}\frac{1}{2}\pi \Sigma \pi ^{\top }<\infty ,
\end{equation}
and by \citet{milgrom.segal.02} $g$ is differentiable everywhere on $(0,\infty )$ with 
\begin{equation}
g^{\prime }(\rho )=-\frac{1}{2}\hat{\pi}(1,\rho )\Sigma \hat{\pi}(1,\rho)^{\top }.  \label{eq: gprime}
\end{equation}
Combination of \eqref{eq: sup pi in A} and \eqref{eq: gprime} proves the right-hand side inequality in \eqref{eq: property g0}. By 
\citet[Theorem 2]{klatte.85}, $\hat{\pi}(1,\rho )$ is a locally Lipschitz-continuous function of $\rho $ and therefore $g^{\prime }$ is also Lipschitz-continuous by \eqref{eq: gprime}. It remains to show that 
\begin{equation}\label{eq: inf rho}
\inf_{\rho \in (0,\gamma)}-g^{\prime }(\rho )>0,
\end{equation} 
which is where the assumption `$\mu _{i}>r$ for some $i\in\{1,\ldots,d\}$' is required. Inequality \eqref{eq: inf rho} holds
through delicate estimates in Lemma~\ref{lem: lambda_minus}.

4) To establish the local Lipschitz property of $x\pi ^{\ast }(t,x)$ note that 
\begin{equation}\label{piast}
x\pi ^{\ast }(t,x)=x\hat{\pi}(1,R(t,x)).
\end{equation}
We have shown in step 3)\ that $\hat{\pi}(1,.)$ is locally Lipschitz-continuous
and since $R(t,x)\in \mathcal{C}^{1,2}( [0,T] \times (0,\infty ))$ the claim
follows. Similar argument applies to $x\widebar{\pi }^{\ast }(t,x)$.

5) For the heuristic strategies $\pi ^{(i)}=\pi^{(i)}(t,x)$, $i\in\{0,1,2,3\}$, the situation is easier because $\pi^{(i)}$ are explicit
functions of $(t,x)$ and the resulting PDE is linear. Logarithmic transformation $z=\ln {x}$ with $u(t,z)=v(t,e^{z})$ transforms the initial value problem \eqref{eq: v(i) IVP} to 
\begin{subequations}
\label{pde_app2_ivp}
\begin{align}
\begin{split}\label{pde_app2}
0=\, &u_{t}^{(i)}+u_{z}^{(i)}\left(ye^{-z}+r+\pi ^{(i)}(\mu -r\mathbf{1})-\frac{1}{2}\pi^{(i)}\Sigma \pi ^{(i)\top }\right)\\
&\hspace{6.5cm}+\frac{1}{2}u_{zz}^{(i)}\pi ^{(i)}\Sigma \pi^{(i)\top }\,,  
\end{split} 
\\
u^{(i)}&(T,z)=\frac{e^{z(1-\gamma )}}{1-\gamma}\,.  \label{pde_app_bc2}
\end{align}
\end{subequations}
By Lemma~\ref{lem: lambda_minus}, equation \eqref{pde_app2} is strictly parabolic for $i\in\{0,1,2,3\}$. Existence of classical 
$\mathcal{C}^{1,2}$ solution follows from standard linear PDE theory 
(\citealp[Theorem III.12.1]{ladyzhenskaya.al.68}, \citealp[Theorem 5.14]{lieberman.96}).

6) In the case $\gamma =1$ we take $U_1(x)=\lim_{\gamma\to 1}\frac{x^{1-\gamma}-1}{1-\gamma}=\ln x$ and the arguments in steps 1)--5)
go through with $u^{(i)}(T,z)=u(T,z)=z$.
\end{proof}

\subsection{Optimality} \label{SS:optimality}

We say $\widebar{\pi }(t,\omega )$ is an admissible control if it is progressively measurable \citep[Definition IV.2.1]{fleming.soner.06}
and $0\leq \widebar{\pi }\mathbf{1}\leq 1-\mathrm{PV}/\widebar{W}$ for $\widebar{W}$ from \eqref{eq: dWbar}, 
\begin{equation}\label{Wbar SDE}
\frac{d\widebar{W}_t}{\widebar{W}_t}=(r+\widebar{\pi }(\mu -r\mathbf{1}))dt+\widebar{\pi }\sigma dB_t.
\end{equation}
Observe that SDE \eqref{Wbar SDE} has a unique strong solution for any progressively measurable $\widebar{\pi }$ with values in the compact set 
$0\leq \widebar{\pi }1\leq 1$ \citep[paragraph after equation IV.2.4]{fleming.soner.06}.

Comparison principle yields the estimate $\left\vert \widebar{v}(t,x)\right\vert \leq e^{C(T-t)}x^{1-\gamma }/\left|1-\gamma\right|$ for 
$\gamma >0,\gamma \neq 1$ and a suitably chosen $C>0$ dependent on $\gamma $. For $\gamma \in (0,1)$ the verification theorem 
\citep[Corollary IV.3.1]{fleming.soner.06} yields directly that $\widebar{\pi }^{\ast }(t,\widebar{W}_t)$ is the optimal Markov control policy. Because Theorem IV.3.1 in \citet{fleming.soner.06} requires the value function to be dominated by a positive power of the endogenous state variable, for $\gamma >1$ we pass to $\widebar{W}^{-1}$ whose SDE reads
\begin{align*}
\widebar{W}_t d\widebar{W}^{-1}_t &=-\widebar{W}_t^{-1}d\widebar{W}_t+\widebar{W}_t^{-2}
d\bigs[\,\widebar{W},\widebar{W}\,\bigs]_t \\
&= \left( \widebar{\pi }\Sigma \widebar{\pi }^{\top }-r-\widebar{\pi }(\mu -r\mathbf{1})\right) dt-\widebar{\pi }\sigma dB_t.
\end{align*}
Hence, by Appendix D in \citet{fleming.soner.06} $\widebar{W}^{-1}$ satisfies for any $m>0$
\begin{equation*}
E\left[ \left( \sup_{0\leq t\leq T}\widebar{W}_{t}^{-1}\right)^{\!\!m}\right]
<\infty .
\end{equation*}
This means $\widebar{v}(t,\widebar{W}_t)$ is a process of class (D) 
\citep[Definition I.1.46]{js.03} and a local supermartingale for any admissible 
strategy $\widebar{\pi }$, hence a supermartingale 
\citep[Appendix 3]{karatzas.kardaras.07}. It is furthermore a local martingale 
and therefore a true martingale \citep[Proposition I.1.47]{js.03}
for the optimal strategy $\widebar{\pi }^{\ast }(t,\widebar{W}_t)$ which
therefore remains an optimal Markov policy also for $\gamma >1$.

Finally, for $\gamma =1$ one has $U_1(x)=\ln x$. By comparison principle, the solution $\widebar{v}(t,x)$ satisfies the estimate 
$\ln x\leq \widebar{v}(t,x)\leq \ln x+C(T-t)$ for a suitably chosen $C>0$. By the It\^o formula 
$$d\ln \widebar{W}_t=\left(r+\widebar{\pi }(\mu -r\mathbf{1})-\frac{1}{2}\widebar{\pi }\Sigma\widebar{\pi }^{\top }\right)dt
+\widebar{\pi }\sigma dB_t$$ 
and therefore $\widebar{v}(t,\widebar{W}_t)$ is a process of class (D). Once again, this implies $\widebar{\pi }^{\ast }(t,\widebar{W}_t)$ is an optimal Markov policy.

The optimality results are summarized in the following theorem.
\begin{theorem} Recall the formal value function $v$ in \eqref{eq: HJB IVP}, the corresponding risk aversion function $R$ in \eqref{eq: R}, and the optimal strategy $\pi^*$ in \eqref{eq: pi*}. The following statements hold.
\begin{enumerate}[label={\rm\arabic{*})}, ref={\rm\arabic{*})}]
\item The solution $v$ of \eqref{eq: HJB IVP} is the value function of the corresponding optimal control problem, that is, it satisfies 
\begin{equation}\label{eq: optimalizacia W}
v(0,x)=\sup_{\pi(\varphi )\in \mathcal{A}_x}E_{t}\left[ \frac{1}{1-\gamma }\left( \varphi _{T}S_{T}\right) ^{1-\gamma }\right] .
\end{equation}
\item For any $x\geq 0$ there is a unique process $W$ satisfying 
\begin{align*}
dW_t &=(y+rW_t)dt+W_t\pi ^{\ast }(t,W_t)\left(\frac{dS_t}{S_t}-r\mathbf{1}dt\right), \\
W_{0} &=x.
\end{align*}
\item The optimal strategy $\varphi $ in \eqref{eq: optimalizacia W} satisfies 
\begin{align*}
\varphi _{t}^{i} &=\pi _{i}^{\ast }(t,W_{t})\frac{W_{t}}{S_{t}^{i}},\qquad i\in\{1,\ldots,d\} \\
\varphi _{t}^{0} &=e^{-rt}W_{t}(1-\pi ^{\ast }(t,W_{t})\mathbf{1}),
\end{align*}
and $\varphi S=W$.
\end{enumerate}
\end{theorem}

\section{Economic analysis and numerical robustness}\label{S:analysis}
\subsection{Illustrative example}\label{SS:illustrative}
Consider the log-normal model of asset returns described in the introduction. Below we present, for illustration, a stylized model using figures broadly consistent with equity and corporate bond markets of developed economies. Numerically, we will take risk-free return of 
$r=1\%$ and two risky assets with drifts $\mu _{1}=2$\% (representing bond returns), $\mu _{2}=10$\% (representing stock returns),
volatilities 5\%, 25\% respectively and correlation -0.05, yielding the covariance matrix 
\begin{equation*}
\Sigma =\left[ 
\begin{array}{cc}
0.0025 & -0.000\,625 \\ 
-0.000\,625 & 0.0625
\end{array}
\right] .
\end{equation*}
The investment horizon has been set to $T=40$ years. We have used the cumulative contribution process $Y_{t}=t/T$ so that the cumulative
contribution is normalized to 1. The present framework provides methodology capable of analyzing and comparing results for various non-linear
contribution profiles, but in the interest of brevity we do not consider them here.

We examine three levels of relative risk aversion; low with $\gamma=2$, moderate $(\gamma =5)$, and high $(\gamma =8)$. We report the
utility of competing strategies both in terms of certainty equivalent wealth and in terms of certainty equivalent internal rate of return.%
\footnote{The certainty equivalent is computed from the formula 
$\mathrm{CE}=\left(E((\varphi _{T}S_{T})^{1-\gamma })\right) ^{1/(1-\gamma )}$. The certainty equivalent internal rate of return is given as the interest rate $\rho$ satisfying $\mathrm{CE}=\int_{0}^{T}e^{\rho (T-t)}y(t)dt$.}

To obtain the function $R(t,x)$ in \eqref{eq: R}, we solve the quasilinear second-order Cauchy problem \eqref{eq: PDE_rho} using the methodology of \citet{kilianova.sevcovic.13}. The solution $\rho(t,z)=R(t,e^z)$ is computed on a Cartesian grid $[0,40]\times [-12,6]$ with temporal step of 0.01 and spatial step of 0.001, with left boundary condition of Robin type and right boundary condition of Neumann type using the built-in Matlab function \texttt{pdepe}. The initial value problem \eqref{eq: HJB IVP} is then solved numerically
by applying the method of characteristics to the linear PDE \eqref{eq: PDE_u_from_rho} with starting values $x \in \{e^{-10}, e^{-9}, \ldots, e^{-5}\}$. These results are subsequently extrapolated to $x=0$ by linear regression in $x$. The optimal control $\pi^*$ is obtained via \eqref{eq: pihat def} and \eqref{piast}.
\subsection{Heuristic strategies $\pi^{(1)}$ and $\pi^{(2)}$}\label{SS:heuristics12}
Let us begin by comparing the performance of the optimal strategy $\pi^{\ast }$, computed numerically as described above, with the
rescaled Samuelson strategy $\pi ^{(0)}$, computed explicitly from equation \eqref{eq: pi(0)}. Table~\ref{tab: CE0-2} shows that 
$\pi ^{\ast }$ significantly outperforms the naive strategy for low and medium levels of risk aversion, while with high risk aversion the outperformance is relatively modest.

\begin{table}[tbp] \centering%
\caption{Certainty equivalents and internal rates of return
 for the heuristic strategies $\pi^{(i)}$, $i=0,1,2$, and the optimal strategy, $\pi^{*}$.}%
\medskip 
\begin{tabular}{ccccccccc}
\hline
$\gamma $ & $\mathrm{CE}^{(0)}$ & $\mathrm{IRR}^{(0)}$ & $\mathrm{CE}^{(1)}$
& $\mathrm{IRR}^{(1)}$ & $\mathrm{CE}^{(2)}$ & $\mathrm{IRR}^{(2)}$ & $\mathrm{CE}^{\ast }$ & $\mathrm{IRR}^{\ast }$ \\ \hline
2 & 2.2584 & 3.64\% & 3.3353 & 5.16\% & 3.3353 & 5.16\% & 3.6501 & 5.50\% \\ 
5 & 1.9720 & 3.08\% & 2.0153 & 3.17\% & 2.0153 & 3.17\% & 2.1782 & 3.49\% \\ 
8 & 1.6872 & 2.42\% & 1.6872 & 2.42\% & 1.7510 & 2.58\% & 1.8164 & 2.74\% \\ 
\hline
\end{tabular}
\label{tab: CE0-2}
\end{table}

To gain better understanding where the outperformance originates from, we first analyze the case $\gamma =8$ where the welfare loss is relatively small. We report in Table~\ref{tab: pi*RRA=8} the optimal portfolio weights $\pi^{\ast}(t,W_{t})$ out of accumulated savings 
(cash in hand) $W_{t}$. The naive weights $\pi ^{(0)}$ in this case coincide with $\pi ^{(1)}$ and are equal to 
$$\frac{(\mu -r\mathbf{1})^\top\Sigma^{-1}}{\gamma} =(54.6\%,18.6\%).$$ 
We observe that for high levels of cash in hand there is good agreement between the optimal and the naive strategy, with the optimal weights tending towards $\pi ^{(0)} = \pi ^{(1)}$ as $W_{t}\rightarrow \infty $. For low level of accumulated savings the difference is substantial, however, with the optimal portfolio being invested fully in stocks while portfolios $\pi ^{(0)} = \pi ^{(1)}$ are not fully invested between stocks and bonds.

\begin{table}[tbp] \centering%
\caption{Optimal strategy $\pi^{*}(t,W_{t})$ as a function of $t$ and $W_{t}$ with $\gamma=8$.}%
\medskip {\footnotesize 
\begin{tabular}{ccccccccccc}
\hline
$W_{t}$ & \multicolumn{2}{c}{$t=0$} & \multicolumn{2}{c}{$t=10$} & 
\multicolumn{2}{c}{$t=20$} & \multicolumn{2}{c}{$t=30$} & \multicolumn{2}{c}{$t=39.975$} \\ \hline  \vspace{-0.1in}\\ 
$10^{-5}$ & 0.000 & 1.000 & 0.000 & 1.000 & 0.000 & 1.000 & 0.000 & 1.000 & 0.000 & 1.000 \\
0.01  & 0.000 & 1.000 & 0.000 & 1.000 & 0.000 & 1.000 & 0.000 & 1.000 & 0.581 & 0.197 \\
0.05  & 0.000 & 1.000 & 0.000 & 1.000 & 0.000 & 1.000 & 0.107 & 0.893 & 0.553 & 0.188 \\
0.1   & 0.000 & 1.000 & 0.000 & 1.000 & 0.158 & 0.842 & 0.451 & 0.549 & 0.550 & 0.187 \\
0.2   & 0.244 & 0.756 & 0.349 & 0.651 & 0.475 & 0.525 & 0.625 & 0.375 & 0.548 & 0.186 \\
0.3   & 0.423 & 0.577 & 0.496 & 0.504 & 0.582 & 0.419 & 0.683 & 0.317 & 0.548 & 0.186 \\
0.5   & 0.569 & 0.431 & 0.614 & 0.386 & 0.668 & 0.332 & 0.730 & 0.270 & 0.547 & 0.186 \\
1     & 0.681 & 0.319 & 0.706 & 0.294 & 0.734 & 0.266 & 0.676 & 0.230 & 0.547 & 0.186 \\
2     & 0.740 & 0.260 & 0.723 & 0.246 & 0.670 & 0.228 & 0.611 & 0.208 & 0.547 & 0.186 \\
20    & 0.569 & 0.193 & 0.564 & 0.192 & 0.559 & 0.190 & 0.553 & 0.188 & 0.546 & 0.186\smallskip \\ \hline
\end{tabular}
}
\label{tab: pi*RRA=8} 
\end{table}

Staying with the case $\gamma = 8$, let us now turn to strategy $\pi ^{(2)}$ which coincides with $\pi ^{(1)}$ for high level of accumulated funds by construction (see eqs. \ref{eq: alpha} and \ref{eq: pi(2)}). Its numerical values, obtained from the explicit 
formula~\eqref{eq: pi(2)}, are displayed in Table~\ref{tab: pi(2)RRA=8}. We observe that $\pi ^{(2)}$ is better behaved for low levels of accumulated funds where it becomes fully invested in bonds and stocks, $\pi^{(2)}(\alpha )=\pi ^{(1)}/(\pi ^{(1)}\mathbf{1})=(74.7\%,25.3\%)$ for $\alpha \leq \pi ^{(1)}\mathbf{1}\approx 73\%$, although the split is such that the funds are far from being fully invested in stocks. We conclude that the welfare difference between the optimal strategy $\pi ^{\ast }$ on the one hand, and the heuristic strategies 
$\pi ^{(0)} = \pi ^{(1)}$ and $\pi ^{(2)}$ on the other hand, reflects the economic value of correct lifestyling strategy at 
\emph{low levels} of accumulated capital.

\begin{table}[tbp] \centering
\caption{Heuristic strategy $\pi^{(2)}(\alpha_t)$ as a function of $t$ and $W_{t}$
with $\gamma=8$.}
\medskip {\footnotesize 
\begin{tabular}{ccccccccccc}
\hline
$W_{t}$ & \multicolumn{2}{c}{$t=0$} & \multicolumn{2}{c}{$t=10$} & 
\multicolumn{2}{c}{$t=20$} & \multicolumn{2}{c}{$t=30$} & \multicolumn{2}{c}{$t=39.975$} \\ \hline  \vspace{-0.1in}\\ 
$10^{-5}$ & 0.747 & 0.253 & 0.747 & 0.253 & 0.747 & 0.253 & 0.747 & 0.253 & 0.747 & 0.253 \\
0.01  & 0.747 & 0.253 & 0.747 & 0.253 & 0.747 & 0.253 & 0.747 & 0.253 & 0.581 & 0.197 \\
0.05  & 0.747 & 0.253 & 0.747 & 0.253 & 0.747 & 0.253 & 0.747 & 0.253 & 0.553 & 0.188 \\
0.1   & 0.747 & 0.253 & 0.747 & 0.253 & 0.747 & 0.253 & 0.747 & 0.253 & 0.550 & 0.187 \\
0.2   & 0.747 & 0.253 & 0.747 & 0.253 & 0.747 & 0.253 & 0.747 & 0.253 & 0.548 & 0.186 \\
0.3   & 0.747 & 0.253 & 0.747 & 0.253 & 0.747 & 0.253 & 0.747 & 0.253 & 0.548 & 0.186 \\
0.5   & 0.747 & 0.253 & 0.747 & 0.253 & 0.747 & 0.253 & 0.747 & 0.253 & 0.547 & 0.186 \\
1     & 0.747 & 0.253 & 0.747 & 0.253 & 0.747 & 0.253 & 0.676 & 0.230 & 0.547 & 0.186 \\
2     & 0.747 & 0.253 & 0.723 & 0.246 & 0.670 & 0.228 & 0.611 & 0.208 & 0.547 & 0.186 \\
20    & 0.569 & 0.193 & 0.564 & 0.192 & 0.559 & 0.190 & 0.553 & 0.188 & 0.546 & 0.186\smallskip \\ \hline
\end{tabular}
}\label{tab: pi(2)RRA=8}
\end{table}

Let us now examine the case $\gamma =2$ whose optimal strategy is displayed in Table~\ref{tab: pi*RRA=2}. We show later in 
Subsection~\ref{SS:pi3} that for high values of cash in hand $W_{t}$ the optimal weights $\pi ^{\ast}(t,W_{t})$ tend to the expression
\begin{equation} \label{eq: pi(2)(1)}
\pi ^{(1)}=\frac{\hat{\pi}}{\gamma}+\zeta \min \left(1-\frac{\hat{\pi}\mathbf{1}}{\gamma},0\right),  
\end{equation}
where 
\begin{equation} \label{eq: zeta}
\zeta =\frac{\mathbf{1}^{\top }\Sigma ^{-1}}{\mathbf{1}^{\top }\Sigma ^{-1}\mathbf{1}}  
\end{equation}
is known as the minimum variance portfolio \citep[eq. 4.8]{ingersoll.87}. In the present example we have 
$\hat{\pi}=(437\%,$ $148\%)$, $\hat{\pi}\mathbf{1}=5.85$, and $\zeta =(95.3\%,\,4.7\%)$. Thus, as the risk aversion falls below $5.85$ there is a strong substitution away from bonds towards stocks. The substitution continues until the risk aversion reaches the level of 
$1.27=\hat{\pi}\mathbf{1}-\hat{\pi}_{1}/\zeta _{1}$ below which all accumulated savings are to be invested in stocks only.

For $\gamma =2$ the portfolio weights $\pi ^{(1)}=\pi ^{(2)}$ are fully invested in proportions 
$$\frac{\hat{\pi}}{2} - \zeta \left(\frac{5.85}{2}-1\right) = (34.9\%,65.1\%)$$ 
while the naive strategy $\pi ^{(0)}$ uses almost the opposite ratio 
$$\pi ^{(0)}=\frac{\hat{\pi}}{\hat{\pi}\mathbf{1}}=(74.7\%,25.3\%).$$ 
Therefore, in addition to the discrepancy between $\pi ^{\ast }$ and $\pi ^{(0)}$ for low values of $W_{t}$ which was present already for 
$\gamma=8$, $\pi ^{(0)}$ faces additional discrepancy of the portfolio weights for high level of accumulated savings. The combined effect makes the strategy $\pi ^{(0)}$ substantially suboptimal for low levels of risk aversion.

\begin{table}[tbp] \centering%
\caption{Optimal strategy $\pi^{*}(t,W_{t})$ as a function of $t$ and $W_{t}$ with $\gamma=2$.}
\medskip {\footnotesize 
\begin{tabular}{ccccccccccc}
\hline
$W_{t}$ & \multicolumn{2}{c}{$t=0$} & \multicolumn{2}{c}{$t=10$} & 
\multicolumn{2}{c}{$t=20$} & \multicolumn{2}{c}{$t=30$} & \multicolumn{2}{c}{$t=39.975$} \\ \hline  \vspace{-0.1in}\\ 
$10^{-5}$ & 0.000 & 1.000 & 0.000 & 1.000 & 0.000 & 1.000 & 0.000 & 1.000 & 0.000 & 1.000 \\
0.01  & 0.000 & 1.000 & 0.000 & 1.000 & 0.000 & 1.000 & 0.000 & 1.000 & 0.311 & 0.689 \\
0.05  & 0.000 & 1.000 & 0.000 & 1.000 & 0.000 & 1.000 & 0.000 & 1.000 & 0.342 & 0.659 \\
0.1   & 0.000 & 1.000 & 0.000 & 1.000 & 0.000 & 1.000 & 0.000 & 1.000 & 0.345 & 0.655 \\
0.2   & 0.000 & 1.000 & 0.000 & 1.000 & 0.000 & 1.000 & 0.000 & 1.000 & 0.347 & 0.653 \\
0.3   & 0.000 & 1.000 & 0.000 & 1.000 & 0.000 & 1.000 & 0.000 & 1.000 & 0.348 & 0.652 \\
0.5   & 0.000 & 1.000 & 0.000 & 1.000 & 0.000 & 1.000 & 0.077 & 0.923 & 0.348 & 0.652 \\
1     & 0.000 & 1.000 & 0.029 & 0.971 & 0.104 & 0.897 & 0.212 & 0.788 & 0.349 & 0.651 \\
2     & 0.147 & 0.853 & 0.180 & 0.820 & 0.225 & 0.775 & 0.281 & 0.720 & 0.349 & 0.651 \\
20    & 0.328 & 0.672 & 0.332 & 0.668 & 0.337 & 0.664 & 0.342 & 0.658 & 0.349 & 0.651\smallskip \\ \hline
\end{tabular}
}
\label{tab: pi*RRA=2}
\end{table}

\subsection{Near-optimal strategy $\pi ^{(3)}$}\label{SS:pi3} 
Previous subsection has highlighted that the optimal trading strategy $\pi^{\ast}$ substantially outperforms the strategy $\pi^{(0)}$ based on mechanical rescaling of fixed Samuelson's portfolio weights $\hat{\pi}$ and, to a lesser extent, also the heuristic strategies $\pi^{(1)}$ and $\pi^{(2)}$. This happens for two reasons: firstly, the relative mix of stocks and bonds in the optimal portfolio varies with the
value of the accumulated savings, moving progressively from stocks to bonds as the value of the savings increases over time. Secondly, for high savings levels the relative weights in stocks and bonds do depend on the risk aversion when risk aversion falls below the sum of credit-unconstrained weights $\hat{\pi}\mathbf{1}$. In this subsection we will examine the `lifestyling' phenomenon in more detail, with the view to providing an analytic approximation of the switching formula.

On inspection of the HJB PDE \eqref{eq: HJBbar}, one notes that the optimal portfolio is given by 
\begin{equation*}
\widebar{\pi }^{\ast }(t,\widebar{W}_{t})
=\underset{\widebar{\pi }\geq 0,\widebar{\pi }\mathbf{1}\leq \alpha _{t}}{\arg \max}\ \widebar{\pi }(\mu -r\mathbf{1})
-\frac{1}{2}\widebar{R}(t,\widebar{W}_{t})\widebar{\pi }\Sigma \widebar{\pi }^{\top},
\end{equation*}
where $\widebar{R}(t,\widebar{W}_{t})$ from equation \eqref{eq: Rbar} is the state-dependent coefficient of relative risk aversion of the indirect utility function and $\alpha _{t}=1-\mathrm{PV}_{t}/\widebar{W}_{t}$. From a purely engineering point of view it makes sense to examine the suboptimal strategy where we replace state-dependent value $\widebar{R}(t,\widebar{W}_{t})$ with the constant 
$\gamma = \widebar{R}(T,\widebar{W}_{T})$, 
\begin{equation}\label{eq: pibar(3)}
\widebar{\pi }^{(3)}(\alpha _{t})
=\arg \max_{\widebar{\pi }\geq 0,\widebar{\pi }\mathbf{1}\leq \alpha _{t}}\widebar{\pi }(\mu -r\mathbf{1})
-\frac{\gamma }{2}\widebar{\pi }\Sigma \widebar{\pi }^{\top } = \hat{\pi}(\alpha _{t},\gamma )=\alpha _{t}\hat{\pi}(1,\alpha _{t}\gamma ).
\end{equation}
In the world with contributions this strategy reads (see Eqs. \ref{eq: pi(3)intro} and \ref{eq: self similar})
\begin{equation}\label{eq: pi(3)}
\pi ^{(3)}(\alpha _{t})=\widebar{\pi }^{(3)}(\alpha _{t})/\alpha _{t}=\hat{\pi}(1,\alpha _{t}\gamma ).  
\end{equation}

The strategies $\pi ^{(i)},\widebar{\pi }^{(i)}$, $i\in\{0,1,2,3\}$ dispense with the need to solve a dynamic programming problem and leave us with a much simpler task of constrained quadratic programming. Whether $\widebar{\pi }^{(3)}$ is a good approximation to the optimal strategy $\widebar{\pi }^{\ast }$ now depends on how close the actual indirect risk aversion $\widebar{R}(t,\widebar{W}_{t})$ is to the fixed value $\gamma $.

\begin{table}[tbp] \centering
\caption{Welfare performance of strategies $\pi^{*}$ and $\pi^{(3)}$ for different levels of risk aversion.}\medskip\, 
\begin{tabular}{ccccc}
\hline
$\gamma $ & $\mathrm{CE}^{\ast }$ & $\mathrm{IRR}^{\ast }$ & $\mathrm{CE}^{(3)}$ & $\mathrm{IRR}^{(3)}$ \\ \hline
2 & {3.6501} & 5.50\% & {3.6496} & 5.50\% \\ 
5 & {2.1782} & 3.49\% & {2.1774} & 3.49\% \\ 
8 & {1.8164} & 2.74\% & {1.8161} & 2.74\% \\ \hline
\end{tabular}%
\label{tab: CE3}%
\end{table}

In Table~\ref{tab: CE3} one observes that the investment strategy $\pi^{(3)}$ is for all practical purposes indistinguishable from the fully
optimal investment $\pi ^{\ast }$ in terms of welfare. On inspection of the portfolio weights in Tables~\ref{tab: pi*RRA=8} and 
\ref{tab: pi(3)RRA=8}, we note the largest discrepancy between the two strategies occurs for $t=0$ at the savings level of $W=0.2$ (recall that $\mathrm{PV}_{0}=0.82$) and it amounts to about 6 percentage points shift towards stocks for the $\pi^{(3)} $ strategy. Thus the near-optimal weights $\pi^{(3)}$ tend to be slightly riskier than the fully optimal investment for middling savings levels. 

Generally speaking, the agreement between $\pi ^{\ast }$ and $\pi ^{(3)}$ is guaranteed to be excellent for very low and very high savings levels, since in the former case both strategies invest the entire cash in hand in stocks, while in the latter case we have already seen the optimal weights of both strategies tend to the value $\pi ^{(3)}(1) = \pi^{(2)}(1) = \pi^{(1)}$ given in \eqref{eq: pi(2)(1)}.

Recall that the heuristic strategies $\pi^{(1)}$, resp. $\pi^{(3)}$, are based on replacing $R(t,W_t)$, resp. $\widebar{R}(t,\widebar{W}_t)$, with $\gamma$. We observe numerically in Table~\ref{tab:Rbar2_8} that $\widebar{R}$ can deviate quite substantially from the constant value $\gamma$. Hence the superior performance of strategy $\pi^{(3)}$ over $\pi^{(1)}$ does not stem from $\widebar{R}$ being
closer to $\gamma$ than $R$ is. Instead, $\pi^{(3)}$ does so well because the largest discrepancy between $\widebar{R}$ and $\gamma$ occurs at low levels of $\alpha_t$ and here both strategies invest everything in stocks. From theory we know $R(t,W_t)\leq\gamma $ 
(see Eq. \ref{eq: rho leq gamma}), which translates to 
$$\widebar{R}(t,\widebar{W}_t)\leq \frac{\gamma}{\alpha_t}.$$ 
On the other hand, the numerical results in Table~\ref{tab:Rbar2_8} suggest $\gamma\leq\widebar{R}(t,\widebar{W}_t) $ for which no theoretical proof is available as yet.

\begin{table}[tbp] \centering
\caption{Near-optimal strategy $\pi^{(3)}(\alpha_t)$ as a function of $t$ and $W_t$ with $\gamma=8$. } \medskip 
{\footnotesize 
\begin{tabular}{ccccccccccc}
\hline
$W_{t}$ & \multicolumn{2}{c}{$t=0$} & \multicolumn{2}{c}{$t=10$} & 
\multicolumn{2}{c}{$t=20$} & \multicolumn{2}{c}{$t=30$} & \multicolumn{2}{c}{$t=39.975$} \\ \hline \vspace{-0.11in}\\ 
$10^{-5}$ & 0.000 & 1.000 & 0.000 & 1.000 & 0.000 & 1.000 & 0.000 & 1.000 & 0.000 & 1.000 \\
0.01  & 0.000 & 1.000 & 0.000 & 1.000 & 0.000 & 1.000 & 0.000 & 1.000 & 0.581 & 0.197 \\
0.05  & 0.000 & 1.000 & 0.000 & 1.000 & 0.000 & 1.000 & 0.084 & 0.916 & 0.553 & 0.188 \\
0.1   & 0.000 & 1.000 & 0.000 & 1.000 & 0.118 & 0.882 & 0.443 & 0.557 & 0.550 & 0.187 \\
0.2   & 0.180 & 0.820 & 0.313 & 0.687 & 0.460 & 0.540 & 0.622 & 0.378 & 0.548 & 0.186 \\
0.3   & 0.387 & 0.613 & 0.476 & 0.524 & 0.574 & 0.426 & 0.682 & 0.318 & 0.548 & 0.186 \\
0.5   & 0.553 & 0.447 & 0.606 & 0.394 & 0.665 & 0.335 & 0.730 & 0.270 & 0.547 & 0.186 \\
1     & 0.678 & 0.323 & 0.704 & 0.296 & 0.734 & 0.267 & 0.676 & 0.230 & 0.547 & 0.186 \\
2     & 0.740 & 0.260 & 0.723 & 0.246 & 0.670 & 0.228 & 0.611 & 0.208 & 0.547 & 0.186 \\
20    & 0.569 & 0.193 & 0.564 & 0.192 & 0.559 & 0.190 & 0.553 & 0.188 & 0.546 & 0.186\smallskip\\ \hline
\end{tabular}
}\label{tab: pi(3)RRA=8}
\end{table}

\begin{table}[tbp]
\caption{Values of $\protect\widebar{R}$ as a function of $t$ and $W_t$ for $\gamma\in\{2,8\}$.}\medskip
\setlength\tabcolsep{2pt}
    \begin{subtable}{.5\textwidth}
    \centering
		{\footnotesize  
		\begin{tabular*}{0.9\linewidth}{@{\extracolsep{\fill}}c c c c c c} \hline
     $W_t\ \backslash\ t$     & 0     & 10    & 20    & 30   & 30.9 \\  \hline \vspace{-0.14in}\\		
	$10^{-5}$ & 5.55  & 4.72  & 3.84  & 2.92  & 2.00 \\
			0.01  & 5.28  & 4.49  & 3.66  & 2.79  & 2.00 \\
			0.05  & 4.51  & 3.86  & 3.19  & 2.50  & 2.00 \\
			0.1   & 3.94  & 3.41  & 2.86  & 2.31  & 2.00 \\
			0.2   & 3.32  & 2.92  & 2.52  & 2.15  & 2.00 \\
			0.3   & 2.98  & 2.67  & 2.35  & 2.08  & 2.00 \\
			0.5   & 2.62  & 2.40  & 2.19  & 2.03  & 2.00 \\
			1     & 2.27  & 2.15  & 2.07  & 2.02  & 2.00 \\
			2     & 2.11  & 2.07  & 2.03  & 2.01  & 2.00 \\
			20    & 2.01  & 2.01  & 2.00  & 2.00  & 2.00 \\ \hline
        \end{tabular*}}
        \caption{$\gamma = 2$}
\label{subtable:gam2}
 \end{subtable}
   \begin{subtable}{.5\textwidth}
   \centering
    {\footnotesize    \begin{tabular*}{0.9\linewidth}{@{\extracolsep{\fill}}c c c c c c}\hline
    $W_t\ \backslash\ t$     & 0     & 10    & 20    & 30   & 30.9 \\  \hline \vspace{-0.14in}\\
    $10^{-5}$ & 13.10 & 11.97 & 10.75 & 9.42  & 8.01 \\
    0.01  & 12.42 & 11.36 & 10.22 & 8.99  & 8.00 \\
    0.05  & 10.68 & 9.85  & 8.98  & 8.22  & 8.00 \\
    0.1   & 9.52  & 8.90  & 8.41  & 8.14  & 8.00 \\
    0.2   & 8.72  & 8.48  & 8.25  & 8.06  & 8.00 \\
    0.3   & 8.55  & 8.35  & 8.16  & 8.03  & 8.00 \\
    0.5   & 8.33  & 8.19  & 8.07  & 8.01  & 8.00 \\
    1     & 8.11  & 8.05  & 8.01  & 8.00  & 8.00 \\
    2     & 8.01  & 8.00  & 8.00  & 8.00  & 8.00 \\
    20    & 8.00  & 8.00  & 8.00  & 8.00  & 8.00 \\ \hline
        \end{tabular*}}
             \caption{$\gamma = 8$}
    \label{subtable:gam8}
    \end{subtable}
		\label{tab:Rbar2_8}
\end{table}

Let us now take a closer look at formula \eqref{eq: pibar(3)}. By completing the square we have 
\begin{equation}\label{eq: pibar(3) LSQ}
\widebar{\pi }^{(3)}(\alpha )=\operatorname*{arg\,min}_{\pi \geq 0,\pi ^{\top }\mathbf{1}\leq \alpha }
\bigs\| \pi \sigma -\gamma^{-1}(\mu -r\mathbf{1})^{\top }\sigma^{-1}\bigs\|^{2}.  
\end{equation}
Since the expression on the right-hand side of \eqref{eq: pibar(3) LSQ} is strictly convex in $\pi$, those constraints in 
\eqref{eq: pibar(3) LSQ} that are not binding can be safely removed and the binding constraints applied with equality. Therefore, if some constraints in \eqref{eq: pibar(3) LSQ} are binding, \eqref{eq: pibar(3) LSQ} is equivalent to 
\begin{equation}
\widebar{\pi }^{(3)}(\alpha )=\operatorname*{arg\,max}_{A_{2}\pi^{\top}\!=\,b_{2}}\Vert
A_{1}\pi ^{\top }-b_{1}\Vert ^{2},  \label{cerny_min}
\end{equation}
where $A_{1}=\sigma ^{\top }$, $b_{1}=\sigma ^{-1}(\mu -r\mathbf{1})/\gamma $ and $A_{2}$, $b_{2}$ represent the binding constraints. Assuming that at least one constraint is binding, the solution of \eqref{cerny_min} is given in \citet[Corollary 4.2]{cerny.09} as 
\begin{equation}
\widebar{\pi }^{(3)}(\alpha )^\top  =A_{1}^{-1}b_{1}+(A_{1}^{\top
}A_{1})^{-1}A_{2}^{\top }(A_{2}(A_{1}^{\top }A_{1})^{-1}A_{2}^{\top
})^{-1}(b_{2}-A_{2}A_{1}^{-1}b_{1})\,.  \label{cerny_LA}
\end{equation}

Suppose that the only binding constraint in \eqref{eq: pi(3)} is $\pi \mathbf{1}=\alpha$. In this case 
$A_{2}=\mathbf{1}^\top=(1,1,\dots ,1)\in \mathbb{R}^{d}$, $b_{2}=\alpha $ and \eqref{cerny_LA} takes the form 
\begin{equation}\label{pialpha}
\widebar{\pi }^{(3)}(\alpha )=\frac{\hat{\pi}}{\gamma}+\zeta\left(\alpha -\frac{\hat{\pi}\mathbf{1}}{\gamma} \right),  
\end{equation}
where $\hat{\pi}$ from equation \eqref{eq: pihat1} represents the optimal unit risk-aversion weights without credit constraint and $\zeta$ from equation \eqref{eq: zeta} is the minimum variance portfolio.

Recall that in our numerical illustration the lifestyling correction vector takes the value $\zeta =(95.3\%,4.7\%)$. For high level of risk aversion $\gamma =8$ the constraint $\pi \mathbf{1}\leq \alpha $ becomes binding below $\hat{\alpha}=5.85/{8}\approx 73\%$. The optimal investment switches 100\% to stocks below $\alpha =15.7\%.$ For low level of risk aversion $\gamma =2$ the constraint 
$\pi \mathbf{1}\leq \alpha $ binds for \emph{all} values of $\alpha \in \lbrack 0,1]$ and the investment switches fully into stocks for all 
$\alpha $ below $63.4\%$. For $\gamma$ below $1.27=\hat{\pi}\mathbf{1}-\hat{\pi}_{1}/\zeta _{1}$ it is optimal to invest the entire cash in hand in stocks \emph{at all times}.

\subsection{Robustness analysis}\label{SS:robustness}

In this subsection we provide compelling evidence that the illustrative example of Subsections~\ref{SS:illustrative}--\ref{SS:pi3} is representative of general results for plausible parameter values. For this purpose, we consider 324 different parametrizations obtained as a 
$3\times 3\times 3\times 3\times 4$ Cartesian product of the following parameter values, 
\begin{subequations}
\label{eq: param}
\begin{align}
\mu _{1}& \in \{1.5\%,2\%,3\%\}, \\
\mu _{2}& \in \{7\%,10\%,13\%\}, \\
\sigma _{1}& \in \{3\%,5\%,7\%\}, \\
\sigma _{2}& \in \{20\%,25\%,30\%\}, \\
\rho & \in \{-20\%,-5\%,5\%,20\%\}.
\end{align}
\end{subequations}
The full set of results is available online in \citet{cerny.melichercik.19.online}. An aggregate summary is reported in 
Table~\ref{tab: robustness}.

We note that strategy $\pi ^{(3)}$ offers and excellent approximation of $\pi ^{\ast }$ across the board. Looking at the detailed results over the 324 individual parametrizations, we observe the largest discrepancies occur for $\rho =-0.2$ and high expected bond return 
$\mu_{1}=0.03$ in combination with low bond return volatility $\sigma _{1}=0.03$.

\begin{table}[tbp] \centering 
\caption{Summary of welfare performance of the optimal strategy $\pi^{*}$ relative to heuristic strategies $\pi^{(i)},\, i\in\{0,1,2,3\}$ over 324 model parametrizations specified in equations (\ref{eq: param}a--e).}\label{tab: robustness}
\medskip 
\begin{tabular}{crccrcrll}
\hline
\vspace*{-0.14in} &  &  &  &  &  &  &  & \\ 
$\gamma $ & \multicolumn{2}{c}{$\frac{\mathrm{CE}^{\ast }-\mathrm{CE}^{(0)}}{\mathrm{CE}^{\ast }}$} 
& \multicolumn{2}{c}{$\frac{\mathrm{CE}^{\ast }-\mathrm{CE}^{(1)}}{\mathrm{CE}^{\ast }}$}
& \multicolumn{2}{c}{$\frac{\mathrm{CE}^{\ast }-\mathrm{CE}^{(2)}}{\mathrm{CE}^{\ast }}$} 
& \multicolumn{2}{c}{$\frac{\mathrm{CE}^{\ast }-\mathrm{CE}^{(3)}}{\mathrm{CE}^{\ast }}$} \\ 
& \multicolumn{1}{c}{avg} & \multicolumn{1}{c}{max} & \multicolumn{1}{c}{avg} & \multicolumn{1}{c}{max} 
& \multicolumn{1}{c}{avg} & \multicolumn{1}{c}{max} & \multicolumn{1}{c}{avg} & \multicolumn{1}{c}{max}\smallskip \\ \hline
\vspace*{-0.14in} &  &  &  &  &  &  &  & \\ 
1 & 52.55\% & 87.78\% & 1.45\% & 6.55\% & 1.45\% & 6.55\% & 0.004\% & 0.083\% \\ 
2 & 34.40\% & 80.07\% & 5.73\% & 12.57\% & 5.52\% & 12.31\% & 0.03\% & 0.19\% \\ 
5 & 13.42\% & 49.17\% & 7.71\% & 14.50\% & 5.81\% & 13.69\% & 0.06\% & 0.39\% \\ 
8 & 8.49\%  & 31.57\% & 6.67\% & 14.52\% & 3.60\% & 12.98\% & 0.05\% & 0.37\%\smallskip \\ \hline
\end{tabular}
\end{table}

\section{Conclusions}\label{S:conclusions}

We have analyzed optimal investment for an individual pension savings plan. As a result of the plan's inability to borrow against future contributions the Samuelson paradigm of investment in constant proportions out of total wealth including current savings and present value of future contributions changes in two important respects. Firstly, for high levels of accumulated savings the relative investment in risky bonds and stocks becomes a function of investor's risk aversion, with strong substitution from bonds towards stocks for lower values of risk aversion. Secondly, for low levels of accumulated savings it becomes optimal to switch entirely to stocks, in an investment pattern known as 
stochastic lifestyling \citep{cairns.al.06}.

Since the computation of the fully optimal strategy is prohibitively technical for practitioners, we have put forward a near-optimal strategy
involving only a static constrained quadratic programme (CQP), easily implementable in a spreadsheet. This CQP strategy is shown to be practically indistinguishable from the optimal investment in terms of its welfare implications. We have provided an explicit formula 
\eqref{pialpha} which helps visualize the lifestyling effect and further lowers the technical barrier towards its implementation.

Three aspects of this research merit further investigation, in our view. As with any suboptimal strategy, it is desirable to have
explicit bounds on the degree of suboptimality. The information relaxation approach of \citet{brown.smith.14} is able to estimate the efficiency loss of suboptimal strategies when the optimal strategy is prohibitively expensive to compute. In our setting the optimal strategy is computationally feasible but perhaps the same approach can produce explicit error bounds. 

Secondly, we have observed in our numerical simulations that the indirect relative risk-aversion coefficient $\widebar{R}$ for the optimal strategy in the Samuelson world \eqref{eq: Rbar} satisfies $\widebar{R}\geq \gamma $, implying that the near-optimal strategy 
$\widebar{\pi}^{(3)}$ is more aggresive than the optimal strategy $\widebar{\pi}^*$. 
It is known from the
comparison principle for parabolic equations that in the world with contributions the corresponding indirect relative risk-aversion coefficient \eqref{eq: R} obeys $R\leq \gamma $, yielding $\widebar{R}\leq \gamma/\alpha_t$. A mathematical proof of 
$\widebar{R}\geq \gamma $ seems rather more elusive at present, cf. \citet{xia.11}. 

Last but not least, the near-optimality result has repercussions for the wider life-cycle portfolio allocation literature 
\citep{ayres.nalebuff.13} and deserves to be explored further in that context.   

\medskip

\noindent \textbf{Acknowledgments} We would like to thank two anonymous referees for their comments. The work of Ale\v{s} 
\v{C}ern\'{y} has been supported by the V\'{U}B Foundation grant `Visiting Professor 2011' . The work of Igor Melicher\v{c}\'{\i}k has been supported by VEGA 1/0251/16 project.

\appendix
\renewcommand{\thesection}{Appendix~\Alph{section}}
\section{Proofs}
\renewcommand{\thesection}{\Alph{section}}
\begin{lemma}
\label{lem: lambda_minus}
Let $\Sigma$ be a positive definite matrix in $\mathbb{R}^{d\times d}$. Under the assumption \eqref{eq: mu > r} function $\hat{\pi}$ from equation \eqref{eq: pihat def} satisfies 
\begin{equation}
0<\inf_{\rho \in (0,\gamma ]}\hat{\pi}(1,\rho )\Sigma \hat{\pi}(1,\rho
)^{\top }\text{.}  \label{eq: lambda minus}
\end{equation}
Moreover, for the investment strategies $\pi ^{(i)},i=0,1,2,3$ one has 
\begin{equation}
0<\inf_{(t,x)\in \lbrack 0,T)\times \mathbb{R}^{+}}\pi ^{(i)}(t,x)\Sigma \pi
^{(i)}(t,x)^{\top }\text{.}  \label{eq: estimate pi app}
\end{equation}
\end{lemma}

\begin{proof} Let $i$ be the index for which $\mu _{i}>r$. Let $c_{i}$ denote the $i$--th diagonal term of the matrix $\Sigma $ and define
\begin{equation*}
q_{\rho }(\pi )=\pi (\mu -r\mathbf{1})-\frac{\rho }{2}\pi \Sigma \pi ^{\top }.
\end{equation*}
Consider $\tilde{\pi}=(0,0,\dots ,\tilde{\pi}_{i},0,\dots ,0)$ with 
\begin{equation*}
\tilde{\pi}_{i}=\min \left( \frac{\mu _{i}-r}{\gamma c_{i}},1\right) >0.
\end{equation*}
For $\frac{\mu _{i}-r}{\gamma c_{i}}\leq 1$ we obtain 
\begin{equation*}
q_{\rho }(\tilde{\pi})=\frac{(\mu _{i}-r)^{2}}{\gamma c_{i}}-\frac{\rho }{2\gamma }\frac{(\mu _{i}-r)^{2}}{\gamma c_{i}}
\geq \frac{1}{2}\frac{(\mu_{i}-r)^{2}}{\gamma c_{i}}.
\end{equation*}
For $(\mu _{i}-r)/(\gamma c_{i})>1$ we have $\tilde{\pi}_{i}=1$ and therefore 
\begin{equation*}
q_{\rho }(\tilde{\pi})=(\mu _{i}-r)-\frac{1}{2}\rho c_{i}\geq (\mu _{i}-r)-\frac{1}{2}\gamma c_{i}\geq \frac{1}{2}(\mu _{i}-r).
\end{equation*}
From the above estimates one obtains
\begin{align*}
\inf_{0<\rho \leq \gamma }\left( \sup_{\tilde{\pi}\mathbf{1}\leq \pi \mathbf{1}\leq 1,\pi \geq 0}q_{\rho }\left( \pi \right) \right) 
&\geq \inf_{0<\rho \leq \gamma }q_{\rho }(\tilde{\pi}) \\
&\geq \min \left(\frac{1}{2}\frac{(\mu _{i}-r)^{2}}{\gamma c_{i}},\frac{1}{2}(\mu _{i}-r)\right) =\delta >0.
\end{align*}
On the other hand, setting $\varepsilon =\frac{1}{2}\frac{\delta }{\mathbf{1}^{\top }\left\vert \mu -r\mathbf{1}\right\vert }>0$ one obtains for all $\rho >0$ 
\begin{equation*}
\sup_{\pi \mathbf{1}\leq \varepsilon ,\pi \geq 0}q_{\rho }\left( \pi \right)
\leq \pi (\mu -r\mathbf{1})\leq \delta /2<\delta .
\end{equation*}
Therefore, arguing by contradiction, the optimal strategy verifies 
\begin{equation*}
\inf_{0<\rho \leq \gamma }\hat{\pi}(1,\rho )\mathbf{1}>\varepsilon ,
\end{equation*}
which in view of the assumed regularity of $\sigma $ guarantees \eqref{eq: lambda minus}. 

It remains to prove \eqref{eq: estimate pi app}. Recall $\pi ^{(0)}$ and $\pi ^{(1)}$ are constant and different from the zero vector therefore the result follows by positive definiteness of $\Sigma $. We have 
$$\pi ^{(2)}=\frac{\pi ^{(1)}}{\max (\pi ^{(1)}\mathbf{1},\alpha _{t})}$$
 and therefore in view of $\alpha (t,x)=x/(\mathrm{PV}_t+x)\leq 1$
\begin{align*}
0 &<\pi ^{(1)}\Sigma \pi ^{(1)\top }\leq \pi ^{(1)}\Sigma \pi ^{(1)\top}
\inf_{(t,x)\in \lbrack 0,T)\times \mathbb{R}_{+}}\frac{1}{\max (\pi ^{(1)}\mathbf{1},\alpha (t,x))} \\
&\leq \inf_{(t,x)\in \lbrack 0,T)\times \mathbb{R}^{+}}\pi^{(2)}(t,x)\Sigma \pi ^{(2)}(t,x)^{\top }.
\end{align*}
Finally, from \eqref{eq: pi(3)} recall $\pi ^{(3)}(t,x)=\hat{\pi}(1,\alpha(t,x)\gamma )$. Inequality \eqref{eq: estimate pi app} now follows from \eqref{eq: lambda minus} because $0\leq \alpha(t,x)\leq 1$.
\end{proof}
\begin{proposition}[\citealt{kilianova.sevcovic.13}]\label{prop: KS13}
Assume $g:\mathbb{R}_{+}\rightarrow \mathbb{R}$ is differentiable, its derivative is Lipschitz-continuous and satisfies inequality 
\eqref{eq: property g0}. Then the following statements hold. 
\begin{enumerate}[label={\rm\arabic{*})}, ref={\rm\arabic{*})}]
\item\label{KS13:1} The Cauchy problem
\begin{equation}\label{eq: PDE_rho}
\begin{split}
\partial _{t}\rho -\partial _{z}^{2}g(\rho )+\partial
_{z}[(y(t)e^{-z}+r)\rho -(1-\rho )g(\rho )] &=0,   \\
\rho (T,z) &=\gamma , 
\end{split}
\end{equation}
has a unique solution $\rho (t,z)$ in $\mathcal{C}^{1,2}\left( [0,T)\times 
\mathbb{R}\right) $. This solution satisfies 
\begin{equation}
0<\rho (t,z)\leq \gamma \text{ on }[0,T)\times \mathbb{R},
\label{eq: rho leq gamma}
\end{equation}
and it is H\"{o}lder-continuous of degree $H^{1+\lambda /2,2+\lambda }$ for
any $0<\lambda <\frac{1}{2}$.
\item\label{KS13:2} For $\rho $ from part~\ref{KS13:1} the linear PDE  
\begin{equation}\label{eq: PDE_u_from_rho}
\begin{split}
u_{t}+u_{z}\left( ye^{-z}+r+g\left( \rho \right) \right) &=0,
 \\
u\left( T,z\right) &=\frac{e^{(1-\gamma )z}}{1-\gamma }. 
\end{split}
\end{equation}
has a unique classical solution $u$.
\item\label{KS13:3} Function $u(t,z)$ from part~\ref{KS13:2} is the unique $\mathcal{C}^{1,2}\left([0,T)\times \mathbb{R}\right)$ solution of the Cauchy problem 
\begin{equation}\label{eq: HJBu}
\begin{split}
u_{t}+u_{z}\left( ye^{-z}+r+g\left( 1-\frac{u_{zz}}{u_{z}}\right) \right)&=0,   \\
u\left( T,z\right) &=\frac{e^{(1-\gamma )z}}{1-\gamma }. 
\end{split}
\end{equation}
\item\label{KS13:4} Conversely, if $u$ denotes the unique classical solution from item~\ref{KS13:3} then $\rho =1-u_{zz}/u_{z}$ is the unique classical solution of \eqref{eq: PDE_rho}.
\end{enumerate}
\end{proposition}

\begin{proof} 
Combine Theorems~3.3 and 5.2 and Proposition~3.4 in \citet{kilianova.sevcovic.13}.
\end{proof}
 

\begin{thebibliography}{26}
\expandafter\ifx\csname natexlab\endcsname\relax\def\natexlab#1{#1}\fi
\providecommand{\url}[1]{\texttt{#1}}
\providecommand{\href}[2]{#2}
\providecommand{\path}[1]{#1}
\providecommand{\DOIprefix}{}
\providecommand{\ArXivprefix}{arXiv:}
\providecommand{\URLprefix}{}
\providecommand{\Pubmedprefix}{pmid:}
\providecommand{\doi}[1]{\href{http://dx.doi.org/#1}{\path{#1}}}
\providecommand{\Pubmed}[1]{\href{pmid:#1}{\path{#1}}}
\providecommand{\bibinfo}[2]{#2}
\ifx\xfnm\relax \def\xfnm[#1]{\unskip,\space#1}\fi
\bibitem[{Ayres and Nalebuff(2013)}]{ayres.nalebuff.13}
\bibinfo{author}{Ayres, I.}, \bibinfo{author}{Nalebuff, B.J.},
  \bibinfo{year}{2013}.
\newblock \bibinfo{title}{Diversification across time}.
\newblock \bibinfo{journal}{Journal of Portfolio Management}
  \bibinfo{volume}{39}, \bibinfo{pages}{73--86}.
\bibitem[{Brown and Smith(2014)}]{brown.smith.14}
\bibinfo{author}{Brown, D.B.}, \bibinfo{author}{Smith, J.E.},
  \bibinfo{year}{2014}.
\newblock \bibinfo{title}{Information relaxations, duality, and convex
  stochastic dynamic programs}.
\newblock \bibinfo{journal}{Operations Research} \bibinfo{volume}{62},
  \bibinfo{pages}{1394--1415}.
\newblock   \DOIprefix\doi{10.1287/opre.2014.1322}.
\bibitem[{Cairns et~al.(2006)Cairns, Blake and Dowd}]{cairns.al.06}
\bibinfo{author}{Cairns, A.J.G.}, \bibinfo{author}{Blake, D.},
  \bibinfo{author}{Dowd, K.}, \bibinfo{year}{2006}.
\newblock \bibinfo{title}{Stochastic lifestyling: optimal dynamic asset
  allocation for defined contribution pension plans}.
\newblock \bibinfo{journal}{Journal of Economic Dynamics \& Control}
  \bibinfo{volume}{30}, \bibinfo{pages}{843--877}.
\newblock   \DOIprefix\doi{10.1016/j.jedc.2005.03.009}.
\bibitem[{{\v{C}}ern\'{y}(2009)}]{cerny.09}
\bibinfo{author}{{\v{C}}ern\'{y}, A.}, \bibinfo{year}{2009}.
\newblock \bibinfo{title}{Characterization of the oblique projector
  {$U(VU)^\dagger V$} with application to constrained least squares}.
\newblock \bibinfo{journal}{Linear Algebra and Its Applications} \bibinfo{volume}{431},
  \bibinfo{pages}{1564--1570}.
\newblock \DOIprefix\doi{10.1016/j.laa.2009.05.025}.
\bibitem[{{\v{C}}ern\'{y} and
  Melicher\v{c}\'{i}k(2019)}]{cerny.melichercik.19.online}
\bibinfo{author}{{\v{C}}ern\'{y}, A.}, \bibinfo{author}{Melicher\v{c}\'{i}k,
  I.}, \bibinfo{year}{2019}.
\newblock \bibinfo{title}{Table of robustness results}.
\newblock
  \bibinfo{howpublished}{\url{https://www.martingales.sk/cerny_melichercik_robustness.html}}.
\newblock \bibinfo{note}{{A}ccessed: 2019-11-21}.
\bibitem[{Fleming and Soner(2006)}]{fleming.soner.06}
\bibinfo{author}{Fleming, W.H.}, \bibinfo{author}{Soner, H.M.},
  \bibinfo{year}{2006}.
\newblock \bibinfo{title}{Controlled {M}arkov processes and viscosity
  solutions}. volume~\bibinfo{volume}{25} of
  \textit{\bibinfo{series}{Stochastic Modelling and Applied Probability}}.
\newblock \bibinfo{edition}{2nd} ed., \bibinfo{publisher}{Springer, New York}.
\bibitem[{Geman et~al.(1995)Geman, El~Karoui and Rochet}]{geman.al.95}
\bibinfo{author}{Geman, H.}, \bibinfo{author}{El~Karoui, N.},
  \bibinfo{author}{Rochet, J.C.}, \bibinfo{year}{1995}.
\newblock \bibinfo{title}{Changes of num\'{e}raire, changes of probability
  measure and option pricing}.
\newblock \bibinfo{journal}{Journal of Applied Probability}
  \bibinfo{volume}{32}, \bibinfo{pages}{443--458}.
\newblock \DOIprefix\doi{10.2307/3215299}.
\bibitem[{Hakansson(1970)}]{hakansson.70}
\bibinfo{author}{Hakansson, N.H.}, \bibinfo{year}{1970}.
\newblock \bibinfo{title}{Optimal investment and consumption strategies under
  risk for a class of utility functions}.
\newblock \bibinfo{journal}{Econometrica} \bibinfo{volume}{38},
  \bibinfo{pages}{587--607}.
\newblock \DOIprefix \doi{10.2307/1912196}.
\bibitem[{Ingersoll(1987)}]{ingersoll.87}
\bibinfo{author}{Ingersoll, J.E.}, \bibinfo{year}{1987}.
\newblock \bibinfo{title}{Theory of Financial Decision Making}.
\newblock Studies in Financial Economics, \bibinfo{publisher}{Rowman \&
  Littlefield}, \bibinfo{address}{Savage}.
\bibitem[{Jacod and Shiryaev(2003)}]{js.03}
\bibinfo{author}{Jacod, J.}, \bibinfo{author}{Shiryaev, A.N.},
  \bibinfo{year}{2003}.
\newblock \bibinfo{title}{Limit Theorems for Stochastic Processes}. volume
  \bibinfo{volume}{288} of \textit{\bibinfo{series}{Comprehensive Studies in
  Mathematics}}.
\newblock \bibinfo{edition}{2nd} ed., \bibinfo{publisher}{Springer-Verlag,
  Berlin}.
\bibitem[{Karatzas and Kardaras(2007)}]{karatzas.kardaras.07}
\bibinfo{author}{Karatzas, I.}, \bibinfo{author}{Kardaras, C.},
  \bibinfo{year}{2007}.
\newblock \bibinfo{title}{The num\'eraire portfolio in semimartingale financial
  models}.
\newblock \bibinfo{journal}{Finance \& Stochastics} \bibinfo{volume}{11},
  \bibinfo{pages}{447--493}.
\newblock \DOIprefix\doi{10.1007/s00780-007-0047-3}.
\bibitem[{Kilianov{\'a} and
  {\v{S}}ev{\v{c}}ovi{\v{c}}(2013)}]{kilianova.sevcovic.13}
\bibinfo{author}{Kilianov{\'a}, S.},
  \bibinfo{author}{{\v{S}}ev{\v{c}}ovi{\v{c}}, D.}, \bibinfo{year}{2013}.
\newblock \bibinfo{title}{A transformation method for solving the
  {H}amilton-{J}acobi-{B}ellman equation for a constrained dynamic stochastic
  optimal allocation problem}.
\newblock \bibinfo{journal}{ANZIAM Journal} \bibinfo{volume}{55},
  \bibinfo{pages}{14--38}.
\newblock \DOIprefix\doi{10.1017/S144618111300031X}.
\bibitem[{Klatte(1985)}]{klatte.85}
\bibinfo{author}{Klatte, D.}, \bibinfo{year}{1985}.
\newblock \bibinfo{title}{On the {L}ipschitz behavior of optimal solutions in
  parametric problems of quadratic optimization and linear complementarity}.
\newblock \bibinfo{journal}{Optimization} \bibinfo{volume}{16},
  \bibinfo{pages}{819--831}.
\newblock \DOIprefix\doi{10.1080/02331938508843080}.
\bibitem[{Ladyzhenskaya et~al.(1968)Ladyzhenskaya, Solonnikov and
  Uraltseva}]{ladyzhenskaya.al.68}
\bibinfo{author}{Ladyzhenskaya, O.A.}, \bibinfo{author}{Solonnikov, V.A.},
  \bibinfo{author}{Uraltseva, N.N.}, \bibinfo{year}{1968}.
\newblock \bibinfo{title}{Linear and quasilinear equations of parabolic type}.
\newblock Translated from the Russian by S. Smith. Translations of Mathematical
  Monographs, Vol. 23, \bibinfo{publisher}{American Mathematical Society,
  Providence, R.I.}
\bibitem[{Lieberman(1996)}]{lieberman.96}
\bibinfo{author}{Lieberman, G.M.}, \bibinfo{year}{1996}.
\newblock \bibinfo{title}{Second order parabolic differential equations}.
\newblock \bibinfo{publisher}{World Scientific Publishing Co., Inc., River
  Edge, NJ}.
\newblock \DOIprefix\doi{10.1142/3302}.
\bibitem[{Milgrom and Segal(2002)}]{milgrom.segal.02}
\bibinfo{author}{Milgrom, P.}, \bibinfo{author}{Segal, I.},
  \bibinfo{year}{2002}.
\newblock \bibinfo{title}{Envelope theorems for arbitrary choice sets}.
\newblock \bibinfo{journal}{Econometrica} \bibinfo{volume}{70},
  \bibinfo{pages}{583--601}.
\newblock \DOIprefix \doi{10.1111/1468-0262.00296}.
\bibitem[{Mulvey et~al.(2008)Mulvey, Simsek, Zhang, Fabozzi and
  Pauling}]{mulvey.al.08}
\bibinfo{author}{Mulvey, J.M.}, \bibinfo{author}{Simsek, K.D.},
  \bibinfo{author}{Zhang, Z.}, \bibinfo{author}{Fabozzi, F.J.},
  \bibinfo{author}{Pauling, W.R.}, \bibinfo{year}{2008}.
\newblock \bibinfo{title}{Assisting defined-benefit pension plans}.
\newblock \bibinfo{journal}{Operations Research} \bibinfo{volume}{56},
  \bibinfo{pages}{1066--1078}.
\newblock \DOIprefix\doi{10.1287/opre.1080.0526}.
\bibitem[{Nutz(2012)}]{nutz.12.mf}
\bibinfo{author}{Nutz, M.}, \bibinfo{year}{2012}.
\newblock \bibinfo{title}{Power utility maximization in constrained exponential
  {L}\'evy models}.
\newblock \bibinfo{journal}{Mathematical Finance} \bibinfo{volume}{22},
  \bibinfo{pages}{690--709}.
\newblock \DOIprefix\doi{10.1111/j.1467-9965.2011.00480.x}.
\bibitem[{Samuelson(1969)}]{samuelson.69}
\bibinfo{author}{Samuelson, P.}, \bibinfo{year}{1969}.
\newblock \bibinfo{title}{Lifetime portfolio selection by dynamic stochastic
  programming}.
\newblock \bibinfo{journal}{The Review of Economics and Statistics}
  \bibinfo{volume}{51}, \bibinfo{pages}{239--246}.
\newblock \DOIprefix \doi{10.2307/1926559}.
\bibitem[{Sodhi(2005)}]{sodhi.05}
\bibinfo{author}{Sodhi, M.S.}, \bibinfo{year}{2005}.
\newblock \bibinfo{title}{Lp modeling for asset-liability management: A survey
  of choices and simplifications}.
\newblock \bibinfo{journal}{Operations Research} \bibinfo{volume}{53},
  \bibinfo{pages}{181--196}.
\newblock \DOIprefix\doi{10.1287/opre.1040.0185}.
\bibitem[{Vila and Zariphopoulou(1997)}]{vila.zariphopoulou.97}
\bibinfo{author}{Vila, J.L.}, \bibinfo{author}{Zariphopoulou, T.},
  \bibinfo{year}{1997}.
\newblock \bibinfo{title}{Optimal consumption and portfolio choice with
  borrowing constraints}.
\newblock \bibinfo{journal}{Journal of Economic Theory} \bibinfo{volume}{77},
  \bibinfo{pages}{402--431}.
\newblock \DOIprefix \doi{10.1006/jeth.1997.2285}.
\bibitem[{Xia(2011)}]{xia.11}
\bibinfo{author}{Xia, J.}, \bibinfo{year}{2011}.
\newblock \bibinfo{title}{Risk aversion and portfolio selection in a
  continuous-time model}.
\newblock \bibinfo{journal}{SIAM Journal on Control and Optimization}
  \bibinfo{volume}{49}, \bibinfo{pages}{1916--1937}.
\newblock \DOIprefix \doi{10.1137/10080871X}.
\bibitem[{Zariphopoulou(1994)}]{zariphopoulou.94}
\bibinfo{author}{Zariphopoulou, T.}, \bibinfo{year}{1994}.
\newblock \bibinfo{title}{Consumption-investment models with constraints}.
\newblock \bibinfo{journal}{{SIAM} Journal on Control and Optimization}
  \bibinfo{volume}{32}, \bibinfo{pages}{59--85}.
\newblock \DOIprefix\doi{10.1137/S0363012991218827}.
\bibitem[{Zhang and Ewald(2010)}]{zhang.ewald.10}
\bibinfo{author}{Zhang, A.}, \bibinfo{author}{Ewald, C.O.},
  \bibinfo{year}{2010}.
\newblock \bibinfo{title}{Optimal investment for a pension fund under inflation
  risk}.
\newblock \bibinfo{journal}{Mathematical Methods of Operations Research}
  \bibinfo{volume}{71}, \bibinfo{pages}{353--369}.
\newblock \DOIprefix \doi{10.1007/s00186-009-0294-5}.

\end{thebibliography}
\end{document}